\documentclass[article,twocolumn,10pt]{IEEEtran}
\usepackage{graphicx}
\usepackage{float}
\usepackage{amsfonts}
\usepackage{amsmath}
\usepackage{amssymb}
\usepackage{cite}
\usepackage{bm}
\usepackage{mathrsfs}
\usepackage{epsfig}
\usepackage{algorithm}
\usepackage{algorithmic}
\usepackage{amsthm}
\usepackage{color}

\newtheorem{lemma}{Lemma}

\newtheorem{proposition}{Proposition}

\newtheorem{definition}{Definition}

\begin{document}
\title{OMP Based Joint  Sparsity Pattern  Recovery  Under Communication Constraints}
\author{\authorblockA{Thakshila Wimalajeewa,   \emph{Member IEEE}  and Pramod K.
Varshney, \emph{Fellow IEEE}}
}

\maketitle \thispagestyle{empty}

\begin{abstract}
We address the problem of joint sparsity pattern recovery   based on low dimensional   multiple
measurement vectors (MMVs) in resource constrained  distributed networks. We assume  that distributed  nodes observe sparse signals which  share the same sparsity pattern  and each node obtains   measurements  via a low dimensional linear operator. When the measurements are  collected at distributed nodes in a communication network, it is often required that  joint sparse recovery  be performed under inherent resource  constraints such as   communication bandwidth and transmit/processing  power.
 We present  two approaches  to take the  communication constraints into account while performing common sparsity pattern recovery. First, we explore the use of  a shared multiple access channel (MAC)  in  forwarding  observations  residing at each node to  a  fusion center.   With MAC, while  the bandwidth requirement does
not depend on the number of nodes, the fusion center has access to  only a  linear combination of the observations. We discuss the conditions under which the common sparsity pattern can be estimated reliably.
Second, we develop  two collaborative algorithms based on Orthogonal Matching Pursuit (OMP), to jointly estimate the common  sparsity pattern  in a decentralized manner with a low communication overhead.
In  the proposed algorithms, each node exploits collaboration among neighboring  nodes by sharing   a small amount of information   for  fusion  at different stages in  estimating  the indices of the true support  in a greedy manner.  Efficiency and effectiveness of the proposed algorithms are demonstrated via simulations  along with a comparison with  the most related existing algorithms considering the trade-off between the performance gain and the communication overhead.
\end{abstract}

\footnotetext[1]{Authors are with the Department of Electrical Engineering and Computer Science, Syracuse University, Syracuse NY 13244. Email: twwewelw@syr.edu, varshney@syr.edu}
\footnotetext[2]{This material is based upon work supported by the National Science
Foundation (NSF) under Grant No. 1307775}
\footnotetext[3]{A part of this work was presented at ICASSP 2013, in Vancouver, Canada with the title 'Cooperative Sparsity Pattern Recovery in Distributed  Networks Via Distributed-OMP'}

\section{Introduction}
 The problem of joint  recovery of sparse signals  based on reduced dimensional multiple measurement vectors (MMVs) is very important  with   a variety of applications including sub-Nyquist sampling of multiband signals  \cite{Mishali1},  multivariate regression \cite{Obozinski1}, sparse signal recovery with multiple sensors \cite{Baron1}, spectrum sensing in cognitive radio networks with multiple cognitive radios \cite{ling1,Zeng1,Bazerque1}, neuromagnetic imaging  \cite{Gorodnitsky2,Cotter1},  and medical imaging \cite{Lee2}.   While MMV problems have been traditionally addressed using sensor
array signal processing techniques, they have attracted considerable recent attention in the context of \emph{compressive sensing} (CS).  When all the low dimensional  multiple measurements are  available at a central processing unit, recent  research efforts have  considered  extending algorithms developed for sparse recovery with a  single measurement vector (SMV) in the context of CS to  joint recovery with MMVs  and discuss  performance guarantees \cite{Tropp4,Tropp3,Cotter1,Chen2,Obozinski2,Wipf2,Eldar4,Eldar1}.

Application  of random  matrices  to obtain low dimensional (compressed)  measurements  at distributed nodes in sensor and cognitive radio networks was explored   in \cite{Baron1,ling1,Zeng1,Bazerque1}.  When low dimensional measurements are collected at multiple nodes in a distributed network, to perform simultaneous  recovery of sparse signals in the optimal way,  it is required that each node transmit measurements and the information regarding the low dimensional sampling operators to  a  central processing unit.
However, this centralized processing   requires large communication burden in terms of communication bandwidth and transmit power. In  sensor and cognitive radio networks, power and communication bandwidth are limited and long range communication may be  prohibitive due to the requirement of  higher  transmit power.  Thus, there is a need for the development of sparse  solutions  when all the   measurements are not available  at a central fusion center.  In this paper, our goal is to investigate  the problem of joint  sparsity pattern recovery with multiple measurement vectors utilizing  minimal amount of transmit power and communication bandwidth.

\subsection{Motivation and our contributions}
We consider a distributed network in which multiple nodes observe  sparse signals that  share the same sparsity pattern. Each node is assumed to make measurements  by employing  a low dimensional linear projection  operator.  The goal is to recover the common sparsity  pattern jointly  under communication constraints.   We present two approaches in this paper. The first  approach  that takes  communication constraints into account in  providing   a centralized solution for sparsity pattern recovery  with MMV   is to transmit only a summary or a function of the observations to a fusion center instead of transmitting raw observations. Use of shared multiple access channels (MAC) in forwarding observations to a fusion center is attractive in certain distributed networks whose bandwidth  does not depend on the number of transmitting nodes \cite{mergen1,li1}.
With the MAC output, the fusion center  does not have  access to individual observation vectors in contrast to employing a bank of parallel channels.  However, due  to the gain in communication bandwidth of MAC over parallel communication, MAC is attractive in applications where  communication bandwidth is limited \cite{mergen1,li1}. We discuss the conditions under which the common sparsity pattern can be estimated reliably. Further, we show that  the performance is  comparable under certain conditions to the case where all the observation vectors  are available at the fusion center separately.

In the second approach, we develop two decentralized    algorithms    with low communication overhead  to estimate  the common  sparsity pattern by appropriate collaboration and fusion among nodes. In addition to the low communication cost, decentralized algorithms are attractive since they   are robust to node and link failures.   In contrast to existing decentralized solutions, we consider    a greedy approach, namely, orthogonal matching pursuit (OMP) for joint sparsity pattern recovery.  Standard OMP with a SMV was introduced in \cite{tropp1}. Simultaneous OMP (S-OMP) for MMVs when all the measurements with a common measurement matrix are available  at a central processing unit   was  discussed in \cite{Tropp2} while in \cite{Baron1} authors extended it to the case where multiple measurements are collected via different measurement matrices.

The two proposed algorithms, called DC-OMP 1 and DC-OMP 2,  perform  sparsity pattern recovery in a greedy manner with collaboration and fusion at each iteration.  In the first  algorithm DC-OMP 1,   single-stage fusion is performed at each iteration via only \emph{one hop communication} in estimating the sparsity pattern. DC-OMP 2, is shown to have performance that is close to S-OMP performed in a centralized manner when a reasonable neighborhood size is employed.    Although DC-OMP 1 lacks in  performance,  the communication overhead required by DC-OMP 1  is  much less  compared to DC-OMP 2. However, in terms of  performance, DC-OMP 1 still provides a significant gain compared to the case where each node performs  standard OMP independently to get an estimate of  the complete support set and then the results are fused  to get a global estimate.     We further show that, in  both  DC-OMP 1 and DC-OMP 2, each node has to perform  less number of iterations  compared to the sparsity index of the sparse signal to reliably estimate the sparsity pattern.  (in both OMP and S-OMP,  at least $k$  iterations are required  where $k$ is the sparsity index of the sparse signals).

This work is based on our preliminary  work  presented in \cite{thakshila_icassp1}. 
In \cite{thakshila_icassp1}, we provided the algorithm development of DC-OMP 1 for common sparsity pattern recovery. In this paper, we significantly extend the work by (i). proposing a second decentralized algorithm for joint sparsity pattern recovery  named DC-OMP 2 which provides a larger  performance gain compared to  DC-OMP 1 at the price of additional  communication cost, (ii). providing theoretical analysis to show the performance gain achieved by DC-OMP 1 and DC-OMP 2 via  collaboration  and fusion compared to  that  without any collaboration among nodes, and  (iii). presenting a new joint sparsity pattern recovery approach  based on the MAC output.

\subsection{Related work}
While most of the work related to joint sparsity pattern recovery with  multiple measurement vectors assumes that the measurements are available at a central fusion center \cite{Tropp3,Tropp4,Cotter1,Kim1,Chen2,Wipf2},  there are several recent research efforts in \cite{ling1,Zeng1,Bazerque1,Ling2,thakshila_icassp1,Rabbat1,Patterson1}, that address the problem of  joint sparse recovery  in a decentralized manner where nodes communicate  with only their neighboring  nodes. The problem of spectrum sensing in multiple cognitive radio networks is cast as a sparse recovery problem in \cite{ling1,Zeng1,Bazerque1} and the authors  propose decentralized solutions based on  optimization techniques. However, performing optimization at sensor nodes may be computationally prohibitive. Further,  in these approaches, each node needs to  transmit the whole estimated vector to its neighbors  at each iteration.    In \cite{Patterson1}, a distributed version of iterative hard thresholding (IHT) is developed for sparse signal recovery in sensor networks in which computational complexity per node is less compared to optimization based  techniques. In a recent work \cite{Sundman_icassp1}, the authors have proposed a decentralized version of a subspace pursuit algorithm for distributed CS. While this paper was under review, another paper appeared  that also  considers a decentralized version of OMP, called  DiOMP,  for sparsity pattern recovery \cite{Sundman_arxiv1} which is different from DC-OMP 1 and  DC-OMP 2 proposed in this paper. In DC-OMP 1 and DC-OMP 2, the goal in performing collaboration is to improve the accuracy via fusion  for  common support recovery compared to the standard OMP, whereas the main idea in DiOMP is to exploit the collaboration to identify the common support when the sparse signals have common plus independent supports among multiple sparse signals.

The rest of the paper is organized as follows.
In Section \ref{motivation}, the observation model, and  background  on sparsity pattern recovery  with a central fusion center   are presented. In Section \ref{MAC_recovery}, the problem of joint sparsity pattern recovery with multiple access communication is discussed.
In Section \ref{decentralized_recovery}, two decentralized greedy  algorithms based on OMP are proposed for joint  sparsity pattern recovery.
    Numerical results are presented in Section \ref{numerical} and concluding remarks are given in Section \ref{conclusions}.

\section{Problem Formulation}\label{motivation}
We consider a scenario where  $L$ nodes in a distributed network observe  sparse signals. More specifically, each node makes measurements of   a  length-$N$ signal denoted by  $\mathbf x_l$ which is assumed to be sparse in the basis $\boldsymbol\Phi$. Further, we assume that all signals $\mathbf x_l$'s for $l=0,1,\cdots, L-1$ are sparse with a common sparsity pattern in the same basis. Let $\mathbf x_l = \boldsymbol\Phi \mathbf s_l$ where $\mathbf s_l$ contains only $k \ll N$ significant coefficients and the locations of the significant coefficients are the same for all $\mathbf s_l$  for $l=0,1,\cdots, L-1$.
A practical situation well modeled by this joint sparsity pattern  in distributed networks is where multiple sensors in a sensor network  acquire the same Fourier sparse signal but with phase shifts and attenuations
caused by signal propagation \cite{Baron1}.  Another application is estimation of  the sparse signal spectrum in a cognitive radio  network with multiple cognitive radios \cite{ling1,Zeng1}.

We assume  that each node obtains a compressed version of the sparse  signal observed via the following linear system,
\begin{eqnarray}
\mathbf y_l&=&\mathbf  A_l \mathbf x_l + \mathbf v_l;\label{obs_1}
\end{eqnarray}
for $l=0,1,\cdots, L-1$
where  $\mathbf A_l$ is the $M \times N$  projection matrix, $M$ is the number of compressive measurements  with $M < N$, and $\mathbf v_l$ is the measurement noise vector  at the $l$-th  node. The noise vector $\mathbf v_l$ is assumed to be iid Gaussian with zero mean vector and the covariance matrix $\sigma_{v}^2 \mathbf I_M$ where $\mathbf I_M$ is the $M\times M $ identity matrix.


\subsection{Sparsity pattern recovery}\label{sparse_recovery}
Our goal is  to jointly estimate the common sparsity pattern of $\mathbf x_l$ for $l=0,1,\cdots, L-1$ based on the underdetermined linear system (\ref{obs_1}). Once the sparse support is known, the problem of estimating the coefficients can be cast as a linear estimation problem and standard techniques such as the least squares method can be employed to estimate the  amplitudes of the non zero coefficients. Further, in certain applications including spectrum sensing in cognitive radio networks, it is sufficient only to identify the sparse support.

 Let $\mathbf B_l = \mathbf A_l \boldsymbol\Phi$ so that $\mathbf y_l$ can be expressed as
 \begin{eqnarray}
 \mathbf y_l = \mathbf B_l \mathbf s_l + \mathbf v_l
   \end{eqnarray}
   for $l=0,1,\cdots, L-1$. Define the support set of $\mathbf s_l$,   $\mathcal U$, to be the set which contains the indices of locations of non zero   coefficients in  $\mathbf s_l$ at $l$-th node:
\begin{eqnarray}
\mathcal U := \{i\in \{0,1,\cdots,N-1\}~|~ \mathbf s_{l} (i) \neq 0 \}
\end{eqnarray}
where $\mathbf s_{l} (i)$ is the $i$-th element of $\mathbf s_{l}$ for $i=0,1,\cdots, N-1$ and $l=0,1,\cdots, L-1$. Then we have $k=|\mathcal U|$ where $|\mathcal U|$ denotes the cardinality of  the set $\mathcal U$. It is noted that the sparse support is denoted by the same notation $\mathcal U$ at each node  due to the assumption of common sparsity pattern.
 Further, let $\boldsymbol \zeta$ be a length-$N$ vector which contains binary elements:
i.e.
\begin{eqnarray}
\boldsymbol \zeta(i) = \left\{
\begin{array}{ccc}
1~ & if~ \mathbf s_l(i) \neq 0\\
0 ~& \mathrm{otherwise}
\end{array}\right.\label{beta_def}
\end{eqnarray}
for any $l$ and  $i=0,1,\cdots, N-1$. In other words, elements in $\boldsymbol \zeta$ are $1$'s corresponding to indices in $\mathcal U$ while all other elements are zeros.  The goal of sparsity pattern recovery is to estimate the set $\mathcal U$ (or the vector $\boldsymbol \zeta$).

\subsection{Sparsity pattern recovery via OMP in a centralized setting}
To make all the observation vectors $\mathbf y_l$ for $l=0,1,\cdots, L-1$  available at a fusion center, a bank of dedicated parallel access channels (PAC) which are independent across nodes has to be used.  With PAC, the observation matrix at the fusion center can be written in the form,
\begin{eqnarray}
\mathbf Y = \tilde {\mathbf S} + \mathbf V\label{obs_PAC_2}
\end{eqnarray}
where $\mathbf Y = [\mathbf y_0, \cdots, \mathbf y_{L-1}]$,  $\tilde {\mathbf S} = [\mathbf B_0\mathbf s_0, \cdots, \mathbf B_{L-1}\mathbf s_{L-1}]$ and $\mathbf V = [\mathbf v_0, \cdots, \mathbf v_{L-1}]$. In the special case where $\mathbf A_l = \mathbf A$ for $l=0,1,\cdots, L-1$, the PAC output can be expressed as,  \begin{eqnarray}
\mathbf Y = \mathbf B \mathbf S + \mathbf V\label{obs_PAC}
\end{eqnarray}
where $\mathbf S = [\mathbf s_{0}, \cdots, \mathbf s_{L-1}]$  and $\mathbf B = \mathbf A \boldsymbol \Phi $.  It is noted that the model (\ref{obs_PAC}) is the widely used form for simultaneous sparse approximation with  MMV which has been studied quite extensively \cite{Cotter1,Tropp3,Tropp4,Chen2,Kim1}.

 \begin{algorithm}
Inputs: $\mathbf y$, $\mathbf B$, $k$
\begin{enumerate}
\item Initialize $t=1$, $\hat{\mathcal U}(0) = \emptyset $, residual vector $\mathbf r_{0} = \mathbf y$
\item Find the index $\lambda(t)$ such that
$
\lambda(t) = \underset{\omega = 0,\cdots,N-1}{\arg~ \max} ~ {|\langle\mathbf r_{t-1}, \mathbf b_{\omega}\rangle|}
$
 \item  Set $\hat{\mathcal U}(t) = \hat{\mathcal U}(t-1) \cup \{\lambda (t)\}$
\item  Compute the projection operator $\mathbf P(t) = \mathbf B(\hat{\mathcal U}(t)) \left(\mathbf B(\hat{\mathcal U}(t)) ^T \mathbf B(\hat{\mathcal U}(t)) \right)^{-1} \mathbf B(\hat{\mathcal U}(t)) ^T$. Update the residual vector:  $\mathbf r_{t} = (\mathbf I - \mathbf P(t))\mathbf y$ (note: $\mathbf B(\hat{\mathcal U}(t))$ denotes the submatrix of $\mathbf B$ in which columns are taken from $\mathbf B$ corresponding to the indices in $\hat{\mathcal U}(t)$)
     \item  Increment $t=t+1$ and go to step 2 if $t\leq k$, otherwise, stop and set $\hat{\mathcal U} = \hat{\mathcal U}(t-1)$

 \end{enumerate}
 \caption{Standard OMP  with SMV}\label{algo0}
 \end{algorithm}
While  optimization techniques for sparse signal recovery such as $l_1$ norm minimization     provide  more promising results in terms  of accuracy, their  computational complexity is higher than that of   greedy techniques such as OMP. With a single measurement vector, the computational complexity with the best known convex optimization based algorithm  is in the order of  $\mathcal O(M^2 N^{3/2})$ while the complexity with OMP is in the order of $\mathcal O(MN  k )$ \cite{Huang_arxiv1}.  Further, it was shown in \cite{Fletcher2} that for sparsity pattern recovery with large Gaussian measurement
matrices in high signal-to-noise ratio (SNR) environments, $l_1$ norm based  Lasso  and OMP have almost identical performance.
 The standard OMP as presented in Algorithm \ref{algo0} is developed in \cite{tropp1} for the SMV case. We omit the subscripts of the vectors and matrices whenever we consider the SMV case. In OMP,  at each iteration, the location of one non zero coefficient of the sparse signal (or the index of a  column of $\mathbf B$ that participates in the measurement vector $\mathbf y$) is  identified. 
 Then the selected column's  contribution is subtracted from $\mathbf y$ and iterations  on the residual are carried out.

\begin{algorithm}
Inputs: Inputs: $\{\mathbf y_l$, $\mathbf B_l\}_{l=0}^{L-1}$, $k$
\begin{enumerate}
\item Initialize $t=1$, $\hat{\mathcal U}(0) = \emptyset $, residual vector $\mathbf r_{l,0} = \mathbf y_l$
\item Find the index $\lambda(t)$ such that
$
\lambda(t) = \underset{\omega = 0,\cdots,N-1}{arg~ max} ~ \sum_{l=0}^{L-1}{{|\langle\mathbf r_{l,t-1}, \mathbf b_{l,\omega}\rangle| }}
$
 \item  Set $\hat{\mathcal U}(t) = \hat{\mathcal U}(t-1) \cup \{\lambda(t)\}$
\item
 Compute the orthogonal projection operator:
 $\mathbf P_l(t) =  \mathbf B_l(\hat{\mathcal U}(t)) \left(\mathbf B_l(\hat{\mathcal U}(t)) ^T \mathbf B_l(\hat{\mathcal U}(t)) \right)^{-1} \mathbf B_l(\hat{\mathcal U}(t)) ^T$\\
Update the residual:  $\mathbf r_{l,t} = (\mathbf I - \mathbf P_l(t))\mathbf y_l$
     \item  Increment $t=t+1$ and go to step 2 if $t\leq k$, otherwise, stop and set $\hat{\mathcal U} = \hat{\mathcal U}(t-1)$

 \end{enumerate}
 \caption{S-OMP  with   different projection matrices}\label{algo_SOMP}
 \end{algorithm}

With MMV when all the measurement vectors are  sampled via the same projection matrix as in (\ref{obs_PAC}),   the support of the sparse signal can be estimated,  using  simultaneous OMP (S-OMP)  algorithm as presented in \cite{Tropp2}. An extension of S-OMP to estimate the  common sparsity pattern with MMV considering different projection matrices as in  (\ref{obs_PAC_2})  is presented in \cite{Baron1} which  is summarized in Algorithm \ref{algo_SOMP}.

To perform  S-OMP with the  MMV model in (\ref{obs_PAC_2}) or (\ref{obs_PAC}),   a  high communication burden to make all the  information available at a central processing unit is required. 
In the following sections, we investigate several schemes to perform sparsity pattern recovery with MMV taking  the communication constraints into account.

\section{Sparsity pattern recovery with multiple access channels (MACs)}\label{MAC_recovery}
The use of a shared MAC for  sending information to the fusion center is attractive in many bandwidth constrained communication networks \cite{mergen1,li1}. In this section, we explore the applicability of the MAC transmission scheme for common sparsity pattern recovery.
With MAC, the received signal vector at the fusion center after  $M$ independent transmissions is given by,
\begin{eqnarray}
\mathbf z= \sum_{l=0}^{L-1} \mathbf y_l = \sum_{l=0}^{L-1} \mathbf B_l \mathbf s_l + \mathbf w\label{MAC_output_1}
\end{eqnarray}
where $\mathbf w  = \sum_{l=0}^{L-1}  \mathbf v_l$. It is further noted that we consider   noise free  observations at the fusion center to make the analysis simpler. With the MAC output (\ref{MAC_output_1}), the fusion center does not have  access to individual observation vectors acquired  at each node. Recovering the set of sparse vectors $\mathbf s_l$'s for $l-0,\cdots, L-1$ from (\ref{MAC_output_1}) becomes a data separation problem \cite{Kutyniok_arxiv1} which can be equivalently represented by the following underdetermined linear system:
\begin{eqnarray}
\mathbf z = [\mathbf B_0| \mathbf B_1| \cdots| \mathbf B_{L-1}] \left[
\begin{array}{cc}
\mathbf s_0\\
.\\
\mathbf s_{L-1}
\end{array}\right] + \mathbf w \label{MAC_separation}
\end{eqnarray}
where the projection matrix $[\mathbf B_0| \mathbf B_1| \cdots| \mathbf B_{L-1}]$ in (\ref{MAC_separation}) is $M\times NL$.
 We can cast the problem of sparse signal recovery based on (\ref{MAC_separation}) as a problem of block sparse signal recovery when all the signals $\mathbf s_l$'s for $l=0,\cdots, L-1$ share the same support as discussed below.
Let $\mathbf b_{li}$ be the $i$-th column vector of the matrix $\mathbf B_l$ for $i=0,1,\cdots, N-1$ and $l=0,1,\cdots,L-1$. Then $\mathbf B_l$ can be expressed by concatenating the column vectors  as  $\mathbf B_l =[\mathbf b_{l0}~\cdots~ \mathbf b_{l(N-1)}]$. Further, let $\mathbf s_l(i)$ denote the $i$-th element of the vector $\mathbf s_l$ for  $i=0,1,\cdots, N-1$ and $l=0,1,\cdots,L-1$. Then, we can express $\sum_{l=0}^{L-1}\mathbf B_l \mathbf s_l $ as
\begin{eqnarray}
\sum_{l=0}^{L-1}\mathbf B_l \mathbf s_l & =& (\mathbf b_{00}~\cdots ~ \mathbf b_{(L-1)0} |~ \mathbf b_{01}~\cdots ~ \mathbf b_{(L-1)1} |\nonumber\\
&~&\cdots | \mathbf b_{0(N-1)}~\cdots ~ \mathbf b_{(L-1)(N-1)})_{M\times LN}\nonumber\\
&~&[\mathbf s_{0}(0)~\cdots \mathbf s_{L-1}(0)|~\mathbf s_{0}(1)~\cdots ~ \mathbf s_{L-1}(1)|\nonumber\\
 &~&\cdots ~ |\mathbf s_{0}(N-1)~\cdots \mathbf s_{L-1}(N-1)]^T_{LN\times 1}\nonumber\\
&=&\mathbf D \mathbf c
\end{eqnarray}
where $\mathbf D =(\mathbf d_0 ~ | \mathbf d_1 ~ | \cdots |\mathbf d_{N-1})$ is a $M\times LN$ matrix where $\mathbf d_j =(\mathbf b_{0j}~\cdots ~ \mathbf b_{(L-1)j})$ for $j=0,1,\cdots, N-1$ and $\mathbf c = [\mathbf r_0~ | \mathbf r_1~|\cdots ~| \mathbf r_{N-1}]^T$ is a $LN\times 1$ vector where $\mathbf r_j = [\mathbf s_{0}(j)~\cdots \mathbf s_{L-1}(j)]$ for $j=0,1,\cdots, N-1$.
Since the sparse vectors $\mathbf s_l$'s for $l=0,1,\cdots, L-1$, share the common support,   $\mathbf c$ can be treated as a block sparse vector with $N$ blocks each of length of $L$ in which  only $k$ blocks are non zero.  Then the MAC output in (\ref{MAC_separation}) can be expressed as,
\begin{eqnarray}
\mathbf z = \mathbf D \mathbf c + \mathbf w\label{z_block}
\end{eqnarray}
where $\mathbf c$ is a block sparse vector  with $N$ blocks of size $L$ each. The capability of recovering $\mathbf c$ from (\ref{z_block}) depends on  the properties of the matrix $\mathbf D$.

In the case where the projection matrix $\mathbf A_l$ at the $l$-th node contains elements drawn from a Gaussian ensemble which are independent over $l=0,1,\cdots, L-1$,  the elements of $\mathbf D$ become  independent realizations of a  Gaussian ensemble.
When $\mathbf D$ contains Gaussian entries, it has been shown in \cite{Eldar3} that a block sparse vector of the form (\ref{z_block}) can be recovered reliably if the following condition
 is satisfied when the noise power  is negligible.
 \begin{definition}[\cite{Eldar3}]{Block RIP:} The matrix $\mathbf D$ satisfies block RIP  with parameter $0< \delta_{k} < 1$ if for every $\mathbf c\in \mathbb{R}^N$ that is block $k$-sparse  we have that,
 \begin{eqnarray}
 (1-\delta_{k})||\mathbf c||_2^2 \leq ||\mathbf D \mathbf c||_2^2 \leq  (1+\delta_{k})||\mathbf c||_2^2.
 \end{eqnarray}
 \end{definition}

 \begin{proposition}[\cite{Eldar3}]\label{prop_block}
 When the matrix $\mathbf D$ satisfies block RIP conditions, the block sparse signal  can be reliably recovered with high probability when
 \begin{eqnarray}
 M \geq \frac{36}{7\delta_0} \left(\ln \left(2{N\choose k}\right) + kL \ln \left(\frac{12}{\delta_0}\right) + t\right)\label{M_MAC}
 \end{eqnarray}
  for $0 < \delta_0 <1$ and $t > 0$.
 \end{proposition}
 It is seen that, the number of compressive measurements  required per node, $M$,  for reliable sparse signal recovery based  on the MAC output is proportional to  $L$.
  Thus, as the network size  $L$  increases,  $M$ should be proportionally increased to recover the set of sparse signals.
    However, when the goal  is not to recover the complete sparse signals, but only the common  sparsity pattern, in the following we show that in the case where $\mathbf A_l = \mathbf A$ for $l=0,\cdots, L-1$, the sparsity pattern can be recovered using the MAC output without increasing $M$. Further, under certain  conditions as discussed below, the sparsity pattern recovery with  the MAC output can be performed with  performance that is  comparable  to the case when  all the observations are available at the fusion center (via PAC).

 When $\mathbf A_l= \mathbf A$ for $l=0,1,\cdots, L-1$, 
 (\ref{MAC_output_1})  reduces to
\begin{eqnarray}
\mathbf z
&=& \mathbf B \bar{\mathbf s} + \mathbf w\label{MAC_same}
\end{eqnarray}
where  $\bar{\mathbf s} = \sum_{l=0}^{L-1} \mathbf s_l$. Since all the signals share a common  sparsity pattern, $\bar{\mathbf s}$ is also a sparse vector with the same sparsity pattern. Thus, the problem of joint sparsity pattern recovery reduces to finding the sparsity pattern of $\bar{\mathbf s}$ based on (\ref{MAC_same}) which is the standard model as considered in CS.
Even though, when  individual vectors are sparse with  significant non zero coefficients, coefficients corresponding to non zero locations of $\bar{\mathbf s}$ can be negligible under certain cases. For example, assume that the elements of non zero coefficients of $\mathbf s_l$ are independent realizations of a  random variable with  mean zero. Then, from the central limit theorem, coefficients at non zero locations of $\bar{\mathbf s}$ may reach zero as $L$ becomes large  resulting in the vector $\bar{\mathbf s}$ with all zeros. However, when the amplitudes of all $\mathbf s_l(k)$'s (for $l=0,\cdots, L-1)$ for  a given $k$ have the same sign, $\bar{\mathbf s}$ becomes a sparse vector with significant non zero coefficients. For example, consider the case where multiple nodes observing the same Fourier sparse signal $\tilde{\mathbf x} = \boldsymbol\Phi\boldsymbol\theta$ with different attenuation factors so that we may express, $\mathbf s_l = \mathbf H_l \boldsymbol\theta$ for $l=0,1,\cdots, L-1$ where $\mathbf H_l$  is a diagonal matrix with positive diagonal elements that correspond to attenuation factors. Since, elements in $\mathbf H_l$ are assumed to be positive, each non zero element in $\mathbf s_l$ has   the same sign as the corresponding non zero element in $\boldsymbol\theta$. Thus, in such cases, $\bar{\mathbf s}$ becomes a sparse vector with significant nonzero coefficients. In the following, we examine the conditions under which  (\ref{MAC_same}) provides  performance that is comparable to  the case where all the observations are available at a fusion center (\ref{obs_PAC}).
\subsection{Necessary conditions for support recovery based on MAC output with any classification rule}
To compare the performance  with MAC and PAC outputs, we find a lower bound on the probability of error in support recovery irrespective of the recovery scheme used for support identification. 
Since we assume that  there are $k$ nonzero elements  in each signal $\mathbf s_l$ for $l=0,\cdots, L-1$, there are  $\Pi  ={N\choose k}$  possible support sets. Selecting the correct support set can be formulated as a multiple hypothesis testing problem with $\Pi$ hypotheses.
Based on Fano's inequality, the probability of error of a  multiple hypothesis testing problem with any decision rule is lower bounded by \cite{cover1}
\begin{eqnarray}
P_e \geq 1 - \frac{\Xi + \log 2}{\log \Pi}\label{pe_lower_bound}
\end{eqnarray}
where  $\Xi $ denotes  the average  \emph{Kullback Leibler} (KL) distance between any pair of densities.  Let $\mathcal D_{M}(p_m(\mathbf z) || p_n(\mathbf z))$ denote the KL distance between pdfs of the MAC output (\ref{MAC_same}) when the sparse supports are  $\mathcal U_m$ and $\mathcal U_n$ respectively, where $m,n = 0,\cdots, \Pi$. Thus,  $\Xi_{MAC} = \frac{1}{\Pi^2} \underset{m,n} {\sum} \mathcal D_{M}(p_m(\mathbf z) || p_n(\mathbf z))$.

When  the projection matrix  $\mathbf A$ is  given, we have $\mathbf z|\mathbf B  \bar{\mathbf s}\sim \mathcal N (\mathbf B \bar {\mathbf s}, \sigma_v^2 L \mathbf I_{M} )$ where $
\mathbf x\sim \mathcal N(\boldsymbol\mu, \boldsymbol\Sigma)$ denotes that the random vector $\mathbf x$ has a joint Gaussian distribution with mean $\boldsymbol\mu$ and the covariance matrix $\boldsymbol\Sigma$.
Then we have,
\begin{eqnarray*}
\mathcal D_{M}(p_m(\mathbf z) || p_n(\mathbf z))= \frac{1}{2\sigma_v^2 L} \left\| \sum_{l=0}^{L-1} (\mathbf B_{\mathcal U_n} {\mathbf s}_{l,\mathcal U_n} - \mathbf B_{\mathcal U_m} {\mathbf s}_{l,\mathcal U_m})\right\|_2^2
\end{eqnarray*}
for $m,n=0,1,\cdots, \Pi$ where $||.||_p$ denotes the $\l_p$ norm, $\mathbf B_{\mathcal U_n}$ is a $M\times k$ submatrix of $\mathbf B$ so that $\mathbf B_{\mathcal U_n}$ contains the columns of $\mathbf B$ corresponding to the column indices in the support set $\mathcal U_n$ for $n=0,1, \cdots, T-1$. We denote ${\mathbf s}_{l,\mathcal U_n}$ to be the $k\times 1$ non sparse vector corresponding to the support  $\mathcal U_n$.

 Similarly,  let  $\mathcal D_{P}(p_m(\mathbf Y) || p_n(\mathbf Y))$  be the KL distance between any two pdfs with the PAC output (\ref{obs_PAC}). Then,  we have $\Xi_{PAC} = \frac{1}{\Pi^2} \underset{m,n} {\sum}\mathcal D_{P}(p_m(\mathbf Y) || p_n(\mathbf Y))$.  With the PAC output, the observation  matrix in (\ref{obs_PAC}) has a matrix variate normal distribution conditioned on $  \mathbf B {\mathbf S}$ with mean $\mathbf B { \mathbf S}$ and covariance matrix $\sigma_v^2 \mathbf I_M \otimes \mathbf I_L$, denoted by,  $\mathbf Y | \mathbf B{ \mathbf S}  \sim \mathcal M \mathcal N_{M,L}(\mathbf B { \mathbf S}, \sigma_v^2 \mathbf I_M \otimes \mathbf I_L)$ where $\otimes$  denotes the Kronecker product. The corresponding KL distance between pdfs when the support sets are $\mathcal U_m$ and $\mathcal U_n$, respectively, with the PAC output is given by,
\begin{eqnarray*}
\mathcal D_{P}(p_m(\mathbf Y) || p_n(\mathbf Y))
 = \frac{1}{2\sigma_v^2 }  \sum_{l=0}^{L-1}  \left\|  (\mathbf B_{\mathcal U_n} {\mathbf s}_{l,\mathcal U_n} - \mathbf B_{\mathcal U_m} {\mathbf s}_{l,\mathcal U_m})\right\|_2^2
\end{eqnarray*}
for $m,n=0,1,\cdots, \Pi$.

\begin{lemma}\label{lemma1}
The  average  KL distance between any pair of densities with MAC and PAC outputs have  the following relationship
\begin{eqnarray}
\Xi_{MAC} \leq \Xi_{PAC}
 \end{eqnarray}
 with the equality only if $\mathbf s_l$'s are the same for all $l=0,\cdots, L-1$.
\end{lemma}
\begin{proof}
The proof follows from Proposition \ref{prop_0} given below.
\end{proof}
\begin{proposition}\label{prop_0}
For given $\mathcal U_m$ and $\mathcal U_n$, we  have,
\begin{eqnarray}
  &~& \frac{1}{L} \left\| \sum_{l=0}^{L-1} (\mathbf B_{\mathcal U_n} {\mathbf s}_{l,\mathcal U_n} - \mathbf B_{\mathcal U_m} {\mathbf s}_{l,\mathcal U_m})\right\|_2^2\nonumber\\
   &\leq& \sum_{l=0}^{L-1}  \left\|  (\mathbf B_{\mathcal U_n} {\mathbf s}_{l,\mathcal U_n} - \mathbf B_{\mathcal U_m} {\mathbf s}_{l,\mathcal U_m})\right\|_2^2\label{condition_KL}
   \end{eqnarray}
with equality only when all $\mathbf s_l$'s for $l=0,\cdots, L-1$ are equal.
\end{proposition}

\begin{proof}
When  $\mathbf s_l = {\mathbf s}$ for $l=0,\cdots, L-1$, it is obvious that the right and left hand sides of (\ref{condition_KL}) are equal. To prove the result when $\mathbf s_l$'s for $l=0,1,\cdots, L-1$ are different,  let $\boldsymbol\beta_{m,n}(l) = \mathbf B_{\mathcal U_n} {\mathbf s}_{l,\mathcal U_n} - \mathbf B_{\mathcal U_m} {\mathbf s}_{l,\mathcal U_m}$ be a length-$M$ vector for given $\mathcal U_m$ and $\mathcal U_n$ and for $l=0,\cdots, L-1$. Then we can write,
\begin{eqnarray}
  &~& \Delta = \frac{1}{L} \left\| \sum_{l=0}^{L-1} (\mathbf B_{\mathcal U_n} {\mathbf s}_{l,\mathcal U_n} - \mathbf B_{\mathcal U_m} {\mathbf s}_{l,\mathcal U_m})\right\|_2^2\nonumber\\
   &-& \sum_{l=0}^{L-1}  \left\|  (\mathbf B_{\mathcal U_n} {\mathbf s}_{l,\mathcal U_n} - \mathbf B_{\mathcal U_m} {\mathbf s}_{l,\mathcal U_m})\right\|_2^2\nonumber\\
   &=& \frac{1}{L}\left(\left\|\sum_{l=0}^{L-1}  \boldsymbol\beta_{m,n}(l) \right\|_2^2 - L \sum_{l=0}^{L-1} \left\| \boldsymbol\beta_{m,n}(l)\right\|_2^2\right). \label{Delta_1}
   \end{eqnarray}
   From Cauchy Schwarz inequality, we have,
   \begin{eqnarray}
   \left\|\sum_{l=0}^{L-1}  \boldsymbol\beta_{m,n}(l) \right\|_2^2&\leq& \left(\sum_{l=0}^{L-1}  \left\| \boldsymbol\beta_{m,n}(l) \right\|_2\right)^2\nonumber\\
   &=& \sum_{l=0}^{L-1} \left\| \boldsymbol\beta_{m,n}(l) \right\|_2^2 \nonumber\\
   &+& \underset{k\neq l} {\sum}\| \boldsymbol\beta_{m,n}(l) \|_2 \| \boldsymbol\beta_{m,n}(k) \|_2.
   \end{eqnarray}
   Thus,  we can write $\Delta$ in (\ref{Delta_1}) as,
   \begin{eqnarray*}
   \Delta &\leq& \frac{-1}{L}\left((L-1) \sum_{l=0}^{L-1} \left\| \boldsymbol\beta_{m,n}(l) \right\|_2^2  \right.\nonumber\\
   &~&\left.+  \underset{k\neq l} {\sum}\| \boldsymbol\beta_{m,n}(l) \|_2 \| \boldsymbol\beta_{m,n}(k) \|_2\right)\nonumber\\
   &=& \frac{-1}{L}\left(\underset{l\neq k, l<k}{\sum}\left( \| \boldsymbol\beta_{m,n}(l) \|_2 - \| \boldsymbol\beta_{m,n}(k) \|_2\right)^2\right)\leq 0
   \end{eqnarray*}
   which completes the proof.
\end{proof}

 From Lemma \ref{lemma1}, it can be seen that, when each node samples via the same projection matrix and the sparse signals $\mathbf s_l$'s are not significantly different from  each other, we may approximate $\Xi_{MAC} \approx \Xi_{PAC}$. Then,  the lower bound on the probability of error (\ref{pe_lower_bound}) in sparsity pattern recovery with the  MAC output becomes  very close to that with access to all observation vectors. From numerical results, we  observe   a comparable performance  as long as  $\mathbf s_l(k)$'s for all $l=0,\cdots, L-1$ for given $k$ have the same sign  even though they are significantly different from each other in amplitude.

\subsubsection{Minimum number of measurements with Gaussian matrices}
When the entries of the measurement matrix $\mathbf A$ are drawn from a Gaussian ensemble with mean zero and variance $1$, the author in \cite{wain2} derived the conditions under which the  probability of error in (\ref{pe_lower_bound}) is bounded away from zero with any recovery technique with a single measurement vector. The main difference in  the sparsity pattern recovery with  the MAC output (\ref{MAC_same}), and  that with the SMV appears in terms of the SNR. Based on the  results in \cite{wain2},  (\ref{MAC_same}) is  capable of recovering the sparsity pattern with any recovery technique if,
\begin{eqnarray}
M > \max \left\{ \frac{\log{N\choose k}}{8k \frac{ \bar{ s}_{\min}^2}{L\sigma_v^2}}, \frac{\log(N-k)}{4 \frac{ \bar{ s}_{\min}^2}{L\sigma_v^2}}\right\}\label{M_lowerBound1}
\end{eqnarray}
where $\bar{ s}_{\min} = \underset{j=0,1,\cdots,N-1}{\min}\bar{\mathbf s}(j)$. With the assumption that $\mathbf s_l(k)$'s for all $l=0,\cdots, L-1$ for a given $k$ have the same sign, we will get $\bar{ s}_{\min}=L~\underset{l, j}{\min}\{\mathbf s_l(j)\}$ where the minimization is over $l=0,\cdots, L-1$ and $j=0,\cdots, N-1$. Define the minimum component SNR to be $\gamma_{c,\min} = \frac{\left(\underset{l, j}{\min}\{\mathbf s_l(j)\}\right)^2}{\sigma_v^2}$. Then, the lower bound on $M$ in (\ref{M_lowerBound1}) can be written as,
\begin{eqnarray}
M > \max \left\{ \frac{\log{N\choose k}}{8k L \gamma_{c,\min}}, \frac{\log(N-k)}{4 L \gamma_{c,\min}}\right\}\label{M_lowerBound2}.
\end{eqnarray}
Thus,  (\ref{M_lowerBound2}) gives the necessary conditions for the sparsity pattern recovery based on the MAC output (\ref{MAC_same}).

\subsubsection{Sparsity pattern recovery based on (\ref{MAC_same}) via OMP}
Based on the MAC output in  (\ref{MAC_same}), the standard OMP as in Algorithm \ref{algo0} can be used to  estimate the sparse support by replacing $\mathbf y $ by $\mathbf z$.


Although the MAC transmission  scheme  saves communication bandwidth compared to PAC,  and provides comparable performance in terms of  common sparsity pattern recovery under certain conditions,  its use is still restrictive due to several reasons. (i). It still requires the knowledge of the measurement matrices at the fusion center which involves a certain communication overhead. (ii). Since the fusion center does not have access to individual observation vectors but only to their  linear combination, the capability of recovering the common  sparsity pattern is limited by the nature of the sparse signals. 
(iv). Extension  to estimate the amplitudes of individual observations is not straight forward since individual measurements are not accessible.

In  certain  communication networks, it is  desirable for each node to have an  estimate of the sparsity pattern of the signal with less communication overhead in contrast to centralized solutions.
Thus, in the following, we consider decentralized algorithms for sparse support recovery where there is no central fusion center to make the final estimation.

\section{Decentralized Sparsity Pattern Recovery via OMP}\label{decentralized_recovery}
In a naive approach to  implement S-OMP  in a distributed manner, each node needs to have the knowledge of the measurement vectors and the projection matrices at every other node, which requires high communication burden and usage of a large memory at distributed nodes. A  simple  approach to minimize the communication overhead compared to employing  S-OMP at each node is to estimate  the common  sparsity pattern independently  at each node based on only its own measurement vector  using standard OMP and exchange the estimated sparsity pattern among the nodes to get a fused estimate.  Although this scheme (we call this scheme as D-OMP in subsequent analysis)  requires each node to transmit only its estimated support set for fusion, it lacks accuracy especially when the number of measurements collected at each node is small. The two  proposed  decentralized algorithms in this paper can be considered as an intermediate approach between these two extreme cases.

When the standard OMP algorithm is  performed at a given  node to estimate the  support of sparse signals with only $k$ non zero elements,  at each iteration, there are $N-k$ possible incorrect indices that can be selected. The probability of selecting  an incorrect index at each iteration increases as the number of measurements at each node, $M$, decreases. In the following, we modify  the standard OMP to exploit collaboration and fusion among distributed nodes in a decentralized manner to reduce the probability of selecting an incorrect index at each iteration compared to that with the standard OMP. 

We introduce  the following additional  notations. Let $\mathcal I = (\mathcal G, \Upsilon)$ be an undirected connected graph with node set $\mathcal G = \{0,1,\cdots, L-1\}$ and edge set $\Upsilon$, where each edge $(i,j)\in \Upsilon$ is an undirected pair of distinct nodes within the one-hop communication range.   Define $\mathcal G_l = \{k| (l,k) \in \Upsilon\}$ to be the set containing the indices of neighboring nodes of the $l$-th  node.
As defined in Section \ref{motivation}, let $\mathcal U$ be the support set which contains the indices of non zero coefficients of the sparse signals and  $\hat{\mathcal U_l}$ be the estimated support set at the $l$-th node for $l=0,1,\cdots, L-1$. $\mathbf B(\mathcal A)$ denotes the submatrix of $\mathbf B$ which has columns of $\mathbf B $ corresponding  to the elements in the set $\mathcal A$ for  $\mathcal A \subset \{0,1,\cdots,N-1\}$. We use $|x|$ to denote the absolute value of a scalar $x$ (it is noted that we use the same notation to denote the cardinality of a set and it should be clear by the context).

\subsection{Algorithm development and strategies: DC-OMP 1}\label{dc-omp}
In the proposed distributed and collaborative OMP algorithm 1 (DC-OMP 1)  stated in Algorithm \ref{algo1},   the goal is to improve the performance by inserting an index fusion stage to the standard OMP algorithm at each iteration. More specifically,  the following two phases are performed during each iteration $t$.

\subsubsection {\bf Phase 1 }
In Phase I, each node estimates an index  based on only its own observations similar to index selection stage in standard OMP (step 2 in Algorithm \ref{algo0}).

\subsubsection{\bf Phase II}
Once  an index $\lambda_l(t)$ is estimated in Phase I, it is exchanged among the neighborhood $\mathcal G_l$.  Subsequently, the $l$-th node will have estimated indices  at all the  nodes in its neighborhood $\mathcal G_l$.     Via fusion, each node selects a set of indices (from $|\mathcal G_l\cup l|$ number of indices)  such that most of the nodes agree upon on the set of indices as detailed next.

\begin{algorithm}
Inputs: $\mathbf y_l$, $\mathbf B_l$,  $\mathcal G_l$,  $k$
\begin{enumerate}
\item Initialize $t=1$, $\hat{\mathcal U}_l(0) = \emptyset $, residual vector $\mathbf r_{l,0} = \mathbf y_l$\\
    \textbf{Phase I}
\item Find the index $\lambda_l(t)$ such that
\begin{eqnarray}
\lambda_l(t) = \underset{\omega = 0,\cdots,N-1}{\arg~ \max} ~ {|\langle\mathbf r_{l,t-1}, \mathbf b_{l,\omega}\rangle|}
\end{eqnarray}
\textbf{Phase II}
\item \emph{Local communication}
\begin{enumerate}
\item Transmit $\lambda_l(t)$ to $\mathcal G_l$
 \item Receive $\lambda_i(t), i\in \mathcal G_l$
\end{enumerate}
 \item Update the index set $\lambda_l^*(t)$, as discussed in subsection \ref{step_3_0}
 \item  Set $\hat{\mathcal U}_l(t) = \hat{\mathcal U}_l(t-1) \cup \{\lambda_l^* (t)\}$ and $l_t = |\hat{\mathcal U}_l(t)|$
\item  Compute the projection operator $\mathbf P_l(t) = \mathbf B_l(\hat{\mathcal U}_l(t)) \left(\mathbf B_l(\hat{\mathcal U}_l(t)) ^T \mathbf B_l(\hat{\mathcal U}_l(t)) \right)^{-1} \mathbf B_l(\hat{\mathcal U}_l(t)) ^T$. Update the residual vector:  $\mathbf r_{l,t} = (\mathbf I - \mathbf P_l(t))\mathbf y_l$
     \item  Increment $t=t+1$ and go to step 2 if $l_t\leq k$, otherwise, stop and set $\hat{\mathcal U}_l = \hat{\mathcal U}_l(t-1)$

 \end{enumerate}
 \caption{DC-OMP 1 for joint sparsity pattern recovery:  Executed at $l$-th node, $\hat{\mathcal U}_l$ contains the estimated indices of the support}\label{algo1}
 \end{algorithm}

\subsubsection{Performing step 4 in Algorithm \ref{algo1}}\label{step_3_0}
For small networks (e.g. cognitive radio networks with few nodes), it is reasonable to assume that $\mathcal G_l\cup l = \mathcal G$. When  $\mathcal G_l\cup l = \mathcal G$, each node has the estimated indices at step 2 in Algorithm \ref{algo1}. Let $\alpha(t) =\{\lambda_l(t)\}_{l=0}^{L-1}$.
  If there are any two indices such  that $\lambda_l(t) = \lambda_m(t)$ for $ l\neq m$, then it is more likely that the corresponding index belongs to the true support.  This is because, the probability that two nodes estimate the same index which does not belong to the true support by performing step 2 is negligible especially when $N-k$ becomes large. The  updated set of indices $\lambda^*_l(t)$   at the $t$-th iteration is computed as below;
\begin{itemize}
\item If there are indices in $\alpha(t)$ with more than one occurrences, such indices are put in the set $\lambda_l^*(t)$ (such that $\lambda_l^*(t) = \{ set~ of ~indices~in~\alpha(t)~ which ~occur ~more~ than~ once \} $.
    \item If  there is no index obtained from step 2  that agrees with one or more nodes so  that all $L$ indices in $\alpha(t) $ are distinct, then select one index from $\alpha(t)$ randomly and put in  $\lambda_l^*(t)$. In this case, to avoid the same index being selected at subsequent iterations, we force  all nodes to use the same index.
\end{itemize}

 Next, consider the case where  $\mathcal G_l\cup l \subset {\mathcal G}$. Then, at the $l$-th node, we have  $\alpha_l(t)=\{\lambda_l(t),\{\lambda_i(t)\}_{i\in \mathcal G_l}\}$. If there are indices in $\alpha_l(t)$ which have more than one  occurrences, those indices are  put in $\alpha_l^*(t)$. Otherwise, if all indices in $\alpha_l(t)$ are distinct, we set $\alpha_l^*(t) = \lambda_l(t)$.  When $\mathcal G_l\cup l \subset {\mathcal G}$, since the $l$-th node does not receive the estimated indices  from the whole network at a given iteration, different nodes may agree upon different sets of  indices at a given iteration.
 When two neighboring nodes agree upon two different sets of indices during  the $t$-th iteration, there is a possibility that one node selects the same index at a later iteration beyond  $t$.  To avoid the  $l$-th node selecting  the same index twice, we  perform an additional step in updating $\lambda_l^*(t)$ compared to the case where $\mathcal G_l\cup l = \mathcal G$; i.e.,  check whether  at least one index in $\alpha_l^*(t)$ belongs to $\hat{\mathcal U}_l(t-1)$. More specifically, if $\alpha_l^*(t) \cup \hat{\mathcal U}_l(t-1) =\hat{\mathcal U}_l(t-1)$, then set $\lambda_l^*(t) = \lambda_l(t)$, otherwise update $\lambda_l^*(t)=\alpha_l^*(t)$. It is noted that, DC-OMP 1 with $\mathcal G_l\cup l \subset {\mathcal G}$ in the worst case (i.e. all the indices in $\alpha_l(t)$ are distinct for all $t$) coincides with the standard OMP.
Moreover, when $|\mathcal G_l \cup l| > k$, it is more likely that the vector $\alpha_l(t)$ has at  least one set of two (or more) indices with the same value, thus, $\lambda_l^*(t)$ is updated appropriately  most of the  time at each iteration.

\subsection{Algorithm development and strategies: DC-OMP 2}
 The proposed distributed and collaborative OMP algorithm 2 (DC-OMP 2) is presented in Algorithm \ref{algo_DC-OMP1}.
Compared to DC-OMP 1, DC-OMP 2 has a measurement fusion stage in Phase I in addition to the index fusion stage in Phase II.

In the case where all the observation vectors  and projection matrices are available at each node, all the nodes in the network  can perform S-OMP as presented  in Algorithm \ref{algo_SOMP}. To perform S-OMP at each node,  the quantity $f_{\omega} = \sum_{l=0}^{L-1}{{|\langle\mathbf r_{l,t-1}, \mathbf b_{l,\omega}\rangle| }}$  for each $\omega=0,\cdots,N-1$ needs to be computed at each iteration where $\mathbf r_{l,t-1}$ and  $\mathbf b_{l,\omega}$ are as defined in Algorithm \ref{algo_SOMP}. In  the first phase of DC-OMP 2, an approximation to this quantity is computed at a given node via only one-hop  \emph{local communication}.

\begin{algorithm}
Inputs: $\mathbf y_l$, $\mathbf B_l$,  $\mathcal G_l$, $\mathcal G$, $k$, $L$
\begin{enumerate}

\item Initialize $t=1$, $\hat{\mathcal U}_l(0) = \emptyset $, residual vector $\mathbf r_{l,0} = \mathbf y_l$\\
    \textbf{Phase I}
   \item  \emph {Local Communication:}
    \begin{enumerate}
    \item Compute $f_{l,\omega}(t) = {|\langle\mathbf r_{l,t-1}, \mathbf b_{l,\omega}\rangle| }$ for $\omega=0,1,\cdots, N-1$ and transmit to neighborhood $\mathcal G_l$
        \item Receive $f_{j,\omega}(t)$ from $\mathcal G_l $ to compute $g_{l,\omega}(t) =f_{l,\omega}(t) + \underset{j\in \mathcal G_l}{\sum} f_{j,\omega}(t)$ for $\omega=0,1,\cdots, N-1$ \\
            \end{enumerate}

\item Find
$
\lambda_l(t) = \underset{\omega = 0,\cdots,N-1}{arg~ max} ~ g_{l,\omega}(t)\\
$\\
\textbf{Phase II}
\item  \emph{ Global  Communication:}
 \begin{enumerate}
 \item Transmit $\lambda_l(t)$ to $\mathcal G$
 \item Receive $\lambda_i(t), i\in \mathcal G\setminus l$

 \end{enumerate}
 \item Find the set $\lambda^*_l (t) $ as in Subsection \ref{step_3_1}
 \item  Update the set of estimated indices:  $\hat{\mathcal U}_l(t) = \hat{\mathcal U}_l(t-1) \cup \{\lambda^*_l(t)\}$, set $l_t = |\mathcal U_l(t)|$
\item Compute the orthogonal projection operator:
 $\mathbf P_l(t) =  \mathbf B_l(\hat{\mathcal U}_l(t)) \left(\mathbf B_l(\hat{\mathcal U}_l(t)) ^T \mathbf B_l(\hat{\mathcal U}_l(t)) \right)^{-1} \mathbf B_l(\hat{\mathcal U}_l(t)) ^T$
\item Update the residual:  $\mathbf r_{l,t} = (\mathbf I - \mathbf P_l(t))\mathbf y_l$
     \item  Increment $t=t+1$ and go to step 2 if $l_t\leq k$, otherwise, stop and set $\hat{\mathcal U}_l = \hat{\mathcal U}_l(t-1)$

 \end{enumerate}
 \caption{DC-OMP 2 for joint sparsity pattern recovery; Executed at $l$-th node, $\hat{\mathcal U}_l$ contains the estimated indices of the support}\label{algo_DC-OMP1}
 \end{algorithm}

\subsubsection {\bf Phase 1 } At the $t$-th iteration, $l$-th node computes
    \begin{eqnarray}
    f_{l,\omega}(t) = {|\langle\mathbf r_{l,t-1}, \mathbf b_{l,\omega}\rangle| }
     \end{eqnarray}
     for $\omega=0,1,\cdots, N-1$ and exchanges  it with  the one-hop neighborhood $\mathcal G_l$. Similarly, every node receives such information from its one-hop neighbors  so that the $l$-th node computes the quantity
     \begin{eqnarray}
     g_{l,\omega}(t) =f_{l,\omega}(t)  + \underset{j\in \mathcal G_l}{\sum} f_{j,\omega}(t)
      \end{eqnarray} for $\omega=0,1,\cdots, N-1$ and $l=0,1,\cdots, L-1$. Then an estimate for the index in the support set,  $\lambda_l(t)$,  at the $l$-th node is computed as $\lambda_l(t) = \underset{\omega = 0,\cdots,N-1}{\arg~ \max} ~ g_{l,\omega}(t)$ as given in Step 3 of Algorithm \ref{algo_DC-OMP1}. 
This  '\emph{initial estimate}', $\lambda_l(t)$ at the $t$-th iteration will be used to get an  updated estimate  in the next phase.

 \subsubsection {\bf Phase II} In Phase II, as in DC-OMP 1, an updated index (or a set of indices) for the true support with higher accuracy compared to  the one that is computed  in Phase I, is obtained by performing  collaboration and fusion  via global communication.  More specifically,  each node transmits  its estimated index to the whole network so that every node in the network receives all the estimated indices denoted by, as before,    $\alpha(t) =\{\lambda_l(t)\}_{l=0}^{L-1}$  from  step 3 of Algorithm \ref{algo_DC-OMP1}.  

\subsubsection{Index fusion in Phase II }\label{step_3_1}
  The index fusion stage is performed as stated in Subsection \ref{step_3_0} for $\mathcal G_l\cup l = \mathcal G$.
It is noted that when $L > k$, it is more likely that the vector $\alpha(t)$ has at  least one set of two (or more) indices with the same value, thus, $\lambda_l^*(t)$ is updated appropriately  most of the  time at each iteration. Since each node has the indices received from all the other nodes in the network, every node has the same estimate for $\mathcal U_l$ when the  algorithm terminates. Further, since more than one index can be selected for the set $\lambda^*_l(t)$ at a given iteration, the algorithm can be terminated in  less than $k$ iterations. In the index fusion stage, the reason for imposing global communication is that after  performing index fusion via global communication, all nodes in the network have the same estimated index set. Thus, the residual computed at step 8 at a given node  corresponds to the same remaining column indices of the dictionary at every node.

\subsection{Performance analysis}
\subsubsection{DC-OMP 1}\label{performance_DCOMP1}
The main difference between DC-OMP 1 and the standard OMP is the additional index fusion stage in Phase II in Algorithm \ref{algo1}. In Phase II in Algorithm \ref{algo1}, two things can happen at a given iteration $t$. (i).  $\lambda_l^*(t) = \lambda_l(t)$ if there are no  indices that occur more than once in $\alpha_l(t)$ or the agreed upon indices are already in $\hat{\mathcal U}_l(t-1)$.  Thus, in this case $Pr(\lambda_l^{\star}(t) \in \mathcal U) =  Pr(\lambda_l(t) \in \mathcal U)$.  In the worst case where $\lambda_l^*(t) = \lambda_l(t)$ for all $t$,  DC-OMP 1 in Algorithm \ref{algo1} coincides with the standard OMP, and the convergence properties are  the same as the standard OMP.  (ii). There can be two or more indices with the same value in $\alpha_l(t)$ and in general there can be multiple such occurrences.  In particular, when the neighborhood size at each node exceeds the sparsity index $k$, $\alpha_l(t)$ contains at least two indices with the  same value as discussed below.

Let $M_1(\delta)$ be the number of measurements required by the standard OMP to estimate the sparsity pattern with probability exceeding $1 -\delta$ for $0 < \delta < 1$ with a given projection matrix. In particular, when the projection matrices  contain realizations of  iid Gaussian random variables,  $M_1(\delta) = ck \log(N/\delta)$ for $\delta\in (0,0.36)$ and constant $c$ \cite{tropp1}. It is noted that we do not restrict our analysis to only  the Gaussian case; it is applicable more generally. Also, in the following analysis, we assume $\mathcal G_l \cup l = \mathcal G $ in DC-OMP 1.

\begin{lemma}
Assume $L > k$. Let  $M_1(\delta)$ be the number of measurements required by the standard OMP to recover the sparsity pattern correctly with probability exceeding $1-\delta$. Then, when $M \geq M_1(\delta)$, DC-OMP 1 at a given node recovers the sparsity pattern with  probability exceeding $1-\frac{1}{(N-2k)}\delta^2$.
\end{lemma}


\begin{proof}
When the index fusion stage  is ignored, DC-OMP 1 coincides with the standard OMP. Thus, when  $M \geq M_1(\delta)$  we have that
\begin{eqnarray}
Pr(\lambda_l(t) \in \mathcal U) \geq 1- \delta. \label{p_lamda}
\end{eqnarray}
for $0< \delta < 1$ where $\lambda_l(t)$ is estimated in step 2 in DC-OMP 1.
Let $u_0, u_1, \cdots, u_{k-1}$ denote the $k$ indices in the true support $\mathcal U$ and $u_k,\cdots,u_{N-1}$ denote the $N-k$ indices in $\mathcal U^c$.
Thus,  the $l$-th node selects one index from the set $\{u_0, \cdots, u_{k-1}\}$ at the $t$-th iteration with probability exceeding $1- \delta$. Since each node in $\mathcal G$ selects an index from $\{u_0, \cdots, u_{k-1}\}$ with probability exceeding $1-\delta$ by performing step $2$ in Algorithm \ref{algo1} when $M\geq M_1(\delta)$, there should be at least two nodes which estimate the same index at a given iteration $t$ when $L >   k$.
When there are at least two indices with the same value in $\alpha_l(t)$, we set $\lambda_l^*(t)=\alpha_l^*(t)$ where $\alpha_l^*(t)$ contains all the indices which have more than one occurrence  in $\alpha_l(t)$.


Let $p_{l,t} = Pr(\lambda_l(t)\in \mathcal U)$ which is determined by the statistical  properties of the projection matrix and the noise at the $l$-th node for $l=0,\cdots, L-1$.
Since no index is selected twice in Algorithm \ref{algo1}, at the $t$-th iteration,  we have,
\begin{eqnarray}
&~&Pr(\lambda_l(t)\in \mathcal U^c)\nonumber\\
& = &\underset{u_m\in \mathcal U^c\setminus \hat{\mathcal U}_{l}(t-1)}\sum Pr(\lambda_l(t) =u_m)  \nonumber\\
&\geq&  |\mathcal U^c\setminus \hat{\mathcal U}_{l}(t-1)| \underset{u_m\in \mathcal U^c\setminus \hat{\mathcal U}_{l}(t-1)}{\min} ~ Prob(\lambda_l(t) =u_m)\nonumber\\
&\geq&  (N-2k) \underset{u_m\in \mathcal U^c\setminus \hat{\mathcal U}_{l}(t-1)}{\min} ~ Prob(\lambda_l(t) =u_m)\label{p_lamda_2}
\end{eqnarray}
 where the  last inequality is because $|\mathcal U^c\setminus \hat{\mathcal U}_{l}(t-1)|\geq (N-2k)$.
Since $p_{l,t} = Pr(\lambda_l(t)\in \mathcal U)$ so that $Pr(\lambda_l(t)\in \mathcal U^c) = 1- p_{l,t}$ and (\ref{p_lamda_2}), we have,
\begin{eqnarray}
\underset{u_m\in \mathcal U^c\setminus \hat{\mathcal U}_{l}(t-1)}{\min} ~ Prob(\lambda_l(t) =u_m) \leq \frac{1}{N-2k}  (1-p_{l,t})
\end{eqnarray}
It is noted that when the absolute values of non significant coefficients of the sparse signals are almost zero,  we can approximate  $\underset{u_m\in \mathcal U^c\setminus \hat{\mathcal U}_{l}(t-1)}{\min} ~ Prob(\lambda_l(t) =u_m) \approx (1-\nu_l)\underset{u_m\in \mathcal U^c\setminus \hat{\mathcal U}_{l}(t-1)}{\max} ~ Prob(\lambda_l(t) =u_m)$ for some $0\leq \nu_l \ll 1$. Then we have,
\begin{eqnarray}
\underset{u_m\in \mathcal U^c\setminus \hat{\mathcal U}_{l}(t-1)}{\max} ~ Prob(\lambda_l(t) =u_m) \leq \frac{1}{k(1+\nu_l)} p_{l,t}
\end{eqnarray}
resulting in
\begin{eqnarray}
Pr(\lambda_l(t)=\tilde\lambda(t))|_{\tilde\lambda(t)\in \mathcal U^c} &\leq& \underset{u_m\in \mathcal U^c\setminus \hat{\mathcal U}_{l}(t-1)}{\max} ~ Prob(\lambda_l(t) =u_m)\nonumber\\
&\leq & \frac{1}{(N-2k)(1-\nu_l)} (1- p_{l,t}) \label{lamda_u}
\end{eqnarray}
for $0\leq \nu_l \ll 1$ and $\tilde\lambda(t))$ can take any value from $u_0, \cdots, u_{N-1}$. Similarly, for any other node in $\mathcal G$, we have,
\begin{eqnarray}
Pr(\lambda_j(t)=\tilde\lambda(t))|_{\tilde\lambda(t)\in \mathcal U^c} &\leq& \underset{u_m\in \mathcal U^c\setminus \hat{\mathcal U}_{j}(t-1)}{\max} ~ Prob(\lambda_j(t) =u_m)\nonumber\\
&\leq & \frac{1}{(N-2k)(1-\nu_j)} (1- p_{j,t}) \label{lamda_u}
\end{eqnarray}
for $0\leq \nu_j \ll 1$ and $j\in \mathcal G_l$.

 For a given node $l$, when $ M \geq M_1(\delta)$ we have $p_{l,t}\geq 1-\delta$. With the assumption that $\nu_j << 1$, we may approximate (\ref{lamda_u}) by,
 \begin{eqnarray}
 Pr(\lambda_j(t)=\tilde\lambda(t))|_{\tilde\lambda(t)\in \mathcal U^c}\leq\frac{\delta}{(N-2k)} \label{lamda_l_more_2}
\end{eqnarray}
when $M\geq M_1(\delta)$.
Thus, whenever there are $n_l(\tilde{\lambda}(t))$ number of nodes in $\alpha_l(t)$ that estimate the same index $\tilde{\lambda}(t)$, the probability that the corresponding index does not belong  to the true support is upper  bound by
\begin{eqnarray}
&~&Pr(\lambda_l^*(t)=\tilde\lambda(t))|_{\tilde\lambda(t)\in \mathcal U^c} \nonumber\\
&=&  \prod_{j=1}^{n_l(\tilde{\lambda}(t))}
Pr(\lambda_j(t)=\tilde\lambda(t))|_{\tilde\lambda(t)\in \mathcal U^c} \leq \left(\frac{\delta}{(N-2k)}\right)^{n_l(\tilde{\lambda}(t))}.  \label{lamda_l_more}
\end{eqnarray}
When $L > k$, we have $|n_l(\tilde{\lambda}(t))|\geq 2$ when $M \geq M_1(\delta)$. Thus, we have,
\begin{eqnarray}
Pr(\lambda_l^*(t)=\tilde\lambda(t))|_{\tilde\lambda(t)\in \mathcal U^c}
&\leq& \left(\frac{\delta}{(N-2k)}\right)^2.
\end{eqnarray}
Thus, taking the union bound we get,
\begin{eqnarray}
Pr(\lambda_l^*(t)\in \mathcal U^c) \leq  \delta^2 \left(\frac{(N-k)}{(N-2k)^2}\right).
\end{eqnarray}
Further, when $k << N $, we may approximate $\frac{k}{(N-2k)^2}\rightarrow 0$. Thus, we have,
\begin{eqnarray}
Pr(\lambda_l^*(t)\in \mathcal U^c) \leq   \frac{\delta^2}{(N-2k)}
\end{eqnarray}completing the proof.
\end{proof}
Thus, with the same number of measurements per node, the index fusion stage in DC-OMP 1 improves the performance of sparsity pattern recovery significantly compared to performing only the standard OMP. The performance gain is illustrated in the numerical results section.

\subsubsection{DC-OMP 2}
In contrast to DC-OMP 1, in DC-OMP 2, the initial estimate in Phase I is obtained via \emph{one-hop communication} and the index fusion in Phase II is performed via   \emph{global communication}. Thus,  the estimates obtained during  both phases are more accurate in DC-OMP 2 than those  in DC-OMP 1.  Further, due to global communication during  Phase II, each node performing  DC-OMP 2 has the same estimate  for the sparsity pattern when the algorithm terminates. When the one-hop neighborhood in Phase I in DC-OMP 2 becomes the whole network, the performance of DC-OMP 2 coincides with  S-OMP being performed at each node. However, due to the index fusion stage in Phase II, DC-OMP 2 terminates with less number of iterations compared to S-OMP.   Further, with a reasonable neighborhood size for local communication in Phase I, DC-OMP 2 provides  performance that is comparable to S-OMP as observed from Simulation results. It is noted that, S-OMP  requires the global knowledge of the observations and the measurement matrices at each node, while  in \emph{global communication} in Phase II in DC-OMP 2, only one index is transmitted  at each iteration.

\subsection{Communication complexity}
To analyze  the communication complexity,  we concentrate on the amount of information to be transmitted by each node and whether that information is transmitted to only one hop neighbors or to the whole network.  Communication complexity of the two proposed decentralized algorithms are compared with two extreme cases: performing S-OMP at each node and performing standard OMP at each  node independently and then fusing  the estimated supports to get a global estimate (D-OMP).

\subsubsection{S-OMP}
When S-OMP as stated in Algorithm \ref{algo_SOMP} is performed at each node, each node  is required to transmit $kN$ messages  to the whole network after $k$ iterations. Thus, the communication burden in terms of the total amount of information to be exchanged in the network is $L(L-1) kN$.

\subsubsection{D-OMP}
In D-OMP, the standard OMP is performed at each node based on its own measurement vector to estimate the support set, and the estimated supports sets are fused via the majority rule to get the final estimate. Each node is required to transmit $k$ indices, thus the communication complexity is in the order of $\mathcal O(k(L-1)L)$.

\subsubsection{DC-OMP 1}
In DC-OMP 1, each node has to communicate only one index to the neighborhood at a given iteration, thus, communication cost  per node after $T_1 \leq k$ iterations is in the order of $\mathcal O(T_1)$. The total number of messages to be transmitted by all the sensors assuming that each node talks to the neighboring nodes one by one, is $\sum_{l=0}^{L-1}|\mathcal G_l| T_1$.   It is also noted that for sufficiently small networks (such as cognitive radios with 10-20 radios) the assumption that $\mathcal G_l\cup l =\mathcal G$ for all $l$ is also reasonable in performing DC-OMP 1 since it is required to transmit only one index to the whole network at each iteration with a maximum of $T_1 < k$ iterations. Then the total number of messages transmitted by all nodes is $(L-1)L T_1$.

\subsubsection{DC-OMP 2}
In DC-OMP 2, each node has to transmit $N$ values to its one-hop  neighbors during  each iteration during Phase I. Thus, after $T_2 \leq k$ iterations, each node  requires $\mathcal O(T_2 N)$ transmissions in the neighborhood. In this phase, it is desirable to have as  small a number of neighbors as possible since $N$ messages are to be transmitted per iteration.   During Phase II, each node exchanges  one  index with  the whole network, thus after $T_2 \leq k$ iterations, each requires $T_2$ transmissions to the whole network. Thus, the total number  of messages transmitted by all nodes is $\sum_{l=0}^{L-1}|\mathcal G_l| NT_2 + L (L-1)T_2$. Compared to performing S-OMP at each node, in DC-OMP 2, communication of length-$N$ messages is restricted to the one-hop neighbors at each node.

It is also noted that the communication complexity above  is computed assuming that each node communicates with  the other nodes one by one. The communication complexity can be further reduced if an efficient broadcast mechanism is used.

Communication complexities in terms of the amount of information to be transmitted   are  summarized in Table \ref{table_comparison}. As defined before, $T_1, T_2 (\leq k)$ denotes the number of iterations required for the algorithm to be terminated with DC-OMP 1 and DC-OMP 2, respectively.

\subsection{Comparison with  optimization based decentralized approach \cite{Zeng1}}
We further compare the communication complexity with the most related  decentralized algorithm for common sparse recovery as considered in \cite{Zeng1}.
As observed in Algorithm 1 in \cite{Zeng1}, to implement the decentralized sparsity pattern  recovery, each node is required to  iteratively solve a  quadratic optimization problem   in an iterative manner to get an estimate of the sparsity pattern. In contrast, in the proposed decentralized algorithms in this paper,  the computational complexity is dominated by the greedy selection stage (which requires less computational complexity compared to performing quadratic  optimization).  Further, in \cite{Zeng1}, at a given iteration, each node is required to communicate a length-$N$ vector to its one hop neighbors, thus the communication complexity in terms of the total number messages to be transmitted by all nodes is $\sum_{l=0}^{L-1}|\mathcal G_l| NT_0$ where $T_0$ is the number of iterations which is different from DC-OMP 1 and DC-OMP 2. It is noted that while DC-OMP 2 has similar communication complexity,  DC-OMP 1 requires much less communication overhead compared to this scheme. Thus, in addition to the computational gain at each node, the communication complexity is less (in DC-OMP 1) or in the same order (in DC-OMP 2) in the proposed algorithms compared to  the optimization based approach presented in \cite{Zeng1}.

\begin{table}\label{table_comparison}
\caption{Comparison of communication complexity of the proposed algorithms  ($T_1, T_2 \leq k <  M  < N$):}
\centering
\centering
\begin{small}
\begin{tabular}{|l|l|l|}
  \hline
  Algorithm   & global commun. & local commun.\\
  \hline
  S-OMP & $L(L-1) kN$  & $-$\\
  D-OMP (with no collab.) & $k(L-1)L$ & $-$\\
  DC-OMP 1 &  $-$ & $\sum_{l=0}^{L-1}|\mathcal G_l| T_1$ \\
DC-OMP 2 &  $L (L-1)T_2$ & $\sum_{l=0}^{L-1}|\mathcal G_l| T_2 N$ \\
    \hline
 \end{tabular}
\end{small}
\end{table}


\section{Simulation Results}\label{numerical}
For simulation purposes, we consider that each sampling  matrix $\mathbf A_l$ is  a random orthoprojector so that $\mathbf A_l \mathbf A_l ^T = \mathbf I_M$ for $l=0,1,\cdots, L-1$.  We illustrate the performance using the probability of exact support recovery, $P_d$,  which is  defined as
\begin{eqnarray}
P_d = Pr (\hat {\boldsymbol \zeta} =  {\boldsymbol \zeta})
\end{eqnarray}
where $\boldsymbol \zeta$ is defined in (\ref{beta_def}) and $\hat {\boldsymbol \zeta}$ contains binary elements $\{0,1\}$ in which $1$'s corresponding to the locations in $\hat{\mathcal U}$.
\begin{figure}[htb]
\centerline{\epsfig{figure=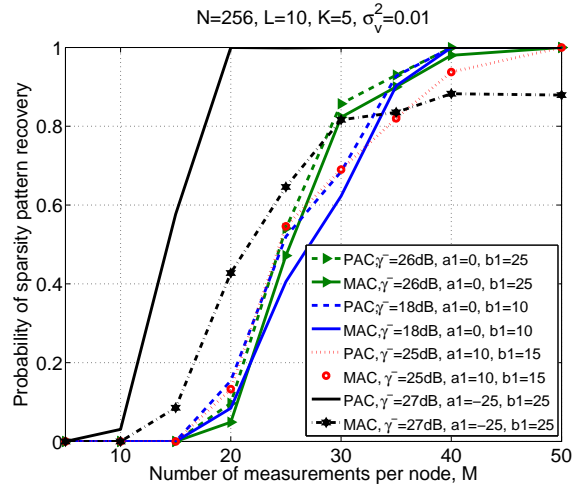,width=8.50cm}}
\caption{Performance of  sparsity pattern recovery via OMP with MAC output (\ref{MAC_same}) and PAC output (\ref{obs_PAC}); $N=256$, $k=5$, $L=10$}
\label{fig_MAC_PAC}
\end{figure}
\subsection{Sparsity pattern recovery with the MAC output}
First,  we illustrate the capability of the MAC transmission scheme (\ref{MAC_same}) in sparsity pattern recovery. In Fig. \ref{fig_MAC_PAC}, we consider the  case where $\mathbf A_l=\mathbf A$ and different choices for the sparse signals $\mathbf s_l$ for $l=0,\cdots, L-1$. With (\ref{MAC_same}), the sparsity pattern recovery  is performed via the standard OMP. We compare the results with the sparsity pattern recovery via S-OMP when all the observations are available at the fusion center (i.e. as in (\ref{obs_PAC})).  In Fig. \ref{fig_MAC_PAC}, we plot the probability of complete support recovery vs the number of measurements per node $M$ when the nonzero elements are drawn from a uniform distribution in the range $[a_1, ~ b_1]$ with different values for $a_1$ and $b_1$. The variance of the measurement noise $\sigma_v^2=0.01$. The average SNR is defined as  $\bar\gamma = \frac{1}{L}\sum_{l=0}^{L-1} \frac{||\mathbf s_l||^2}{N \sigma_v^2}$. In Fig. \ref{fig_MAC_PAC}, the dimension  of the sparse signals is taken as $N=256$,  the sparsity index $k=5$ and the number of nodes  $L=10$. It is observed that when $a_1$ and $b_1$ are such that $\mathbf s_l$'s have the same sign, the performance gap between the MAC  and the PAC outputs in common sparsity pattern recovery  is not significant even  though their amplitudes can differ significantly. In other words, when the fusion center does not have separate observations vectors but has only their  linear combination,  the sparsity pattern recovery can still  be performed reliably (with almost the same performance as the case where all observations are available)   when the coefficients of the sparse signals have the same sign.  However, as claimed in Section \ref{MAC_recovery}, the performance of the sparsity pattern recovery based on  the MAC output (\ref{MAC_same}) is not comparable with the performance with the PAC output (\ref{obs_PAC}) when the nonzero elements of $\mathbf s_l$'s have zero mean (i.e. $a_1=-25$ and $b_1=25$ in Fig. \ref{fig_MAC_PAC}).

\begin{figure}[htb]
\centerline{\epsfig{figure=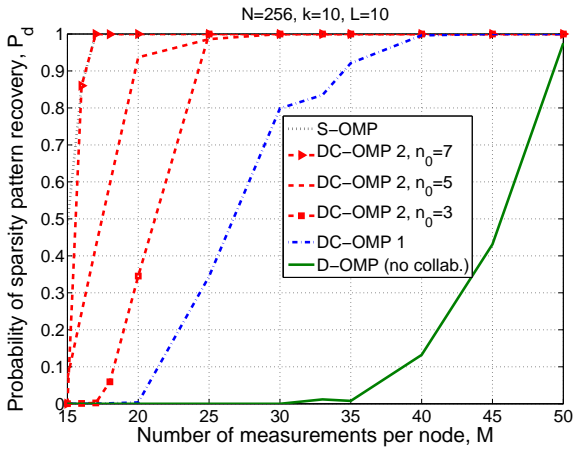,width=8.50cm}}
\caption{Performance of the decentralized sparsity pattern recovery with Algorithms \ref{algo_DC-OMP1} and \ref{algo1}; $N=256$, $k=10$, $L=10$, $\bar\gamma  = \frac{1}{L}\sum_{l=0}^{L-1} \frac{||\mathbf s_l||^2}{N \sigma_v^2} = 28 dB$}
\label{fig_1}
\end{figure}

\begin{figure}[htb]
\centerline{\epsfig{figure=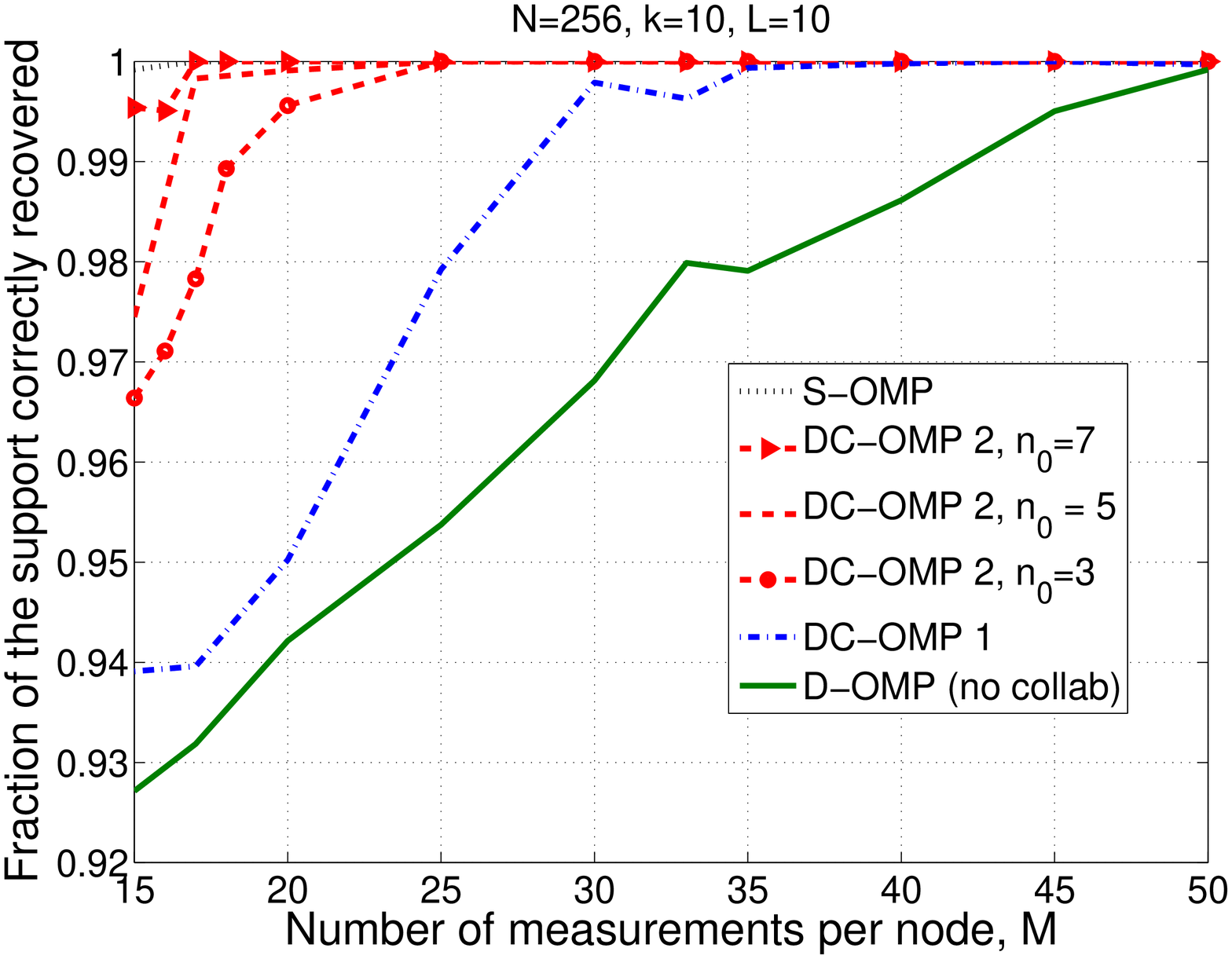,width=8.50cm}}
\caption{Fraction  of the support correctly recovered form  Algorithms \ref{algo_DC-OMP1} and \ref{algo1}; $N=256$, $k=10$, $L=10$, $\bar\gamma  = \frac{1}{L}\sum_{l=0}^{L-1} \frac{||\mathbf s_l||^2}{N \sigma_v^2} = 28 dB$}
\label{fig_2}
\end{figure}
\subsection{Performance of DC-OMP 1 and DC-OMP 2 }
To compare the performance of the proposed  Algorithms \ref{algo1} and \ref{algo_DC-OMP1}  with other comparable approaches, we consider two existing  benchmark cases. (i). Distributed OMP with no collaboration (D-OMP), in this case, each node performs standard OMP in Algorithm \ref{algo0} independently  to obtain the support set estimate $\hat{\mathcal U}_l$. To fuse the estimated  support sets, $\hat{\mathcal U}_l$'s,  at individual nodes,  each node transmits indices in $\hat{\mathcal U}_l$ to the rest of the nodes, and employs  a majority rule based fusion scheme to obtain a final estimate  $\hat{\mathcal U}$.
(ii). S-OMP \cite{Tropp2}: S-OMP algorithm as stated in Algorithm  \ref{algo_SOMP} is carried out at each node where each node has the access to all  the information residing at at every other node.

In Figures \ref{fig_1} and \ref{fig_2}, we let $N=256$, $k=10$ and $L=10$. Non zero coefficients of sparse vectors are generated  as realizations of a uniform random variable with the support $[10,15]$ and the variance of the measurement noise $\sigma_v^2=0.01$. The average SNR  is taken as $\bar\gamma  = 28 dB$.
In Figures \ref{fig_1}-\ref{fig_03}, for DC-OMP 1, we consider  the case where $\mathcal G_l\cup l =\mathcal G$ for $l=0,1,\cdots, L-1$. Thus, it shows the maximum performance achievable via DC-OMP 1 and also each node has the  same estimated support set  based on Algorithm \ref{algo1}.  For DC-OMP 2 in Algorithm \ref{algo_DC-OMP1}, we consider different  neighborhood sizes for local communication phase at each node as $|\mathcal G_l|=n_0 = 3,5, 7$ for $l=0,1,\cdots,L-1$.  Since in DC-OMP 2, index fusion stage is performed via global communication, each node has the same estimate after $T_2$ iterations.

We plot the probability of correctly recovering the full  support set, $P_d$, and the fraction  of the support set that is estimated correctly vs $M$ in Figures \ref{fig_1} and \ref{fig_2}, respectively. Results are obtained  by performing $10^4$ runs and averaging over $50$ trials.  It can be seen from Figures \ref{fig_1} and \ref{fig_2} that the two  proposed  decentralized algorithms outperform  D-OMP with no collaboration. The performance of DC-OMP 2 gets closer to S-OMP as  the neighborhood size for local communication phase increases.
While  a considerable performance gain over D-OMP with no collaboration is achieved,  DC-OMP 1 still has a certain performance gap compared to S-OMP and DC-OMP 2. It is noted that,   DC-OMP 2 performs measurement fusion within a neighborhood prior to  index fusion based on the estimated indices at all nodes in the network. On the other hand, DC-OMP 1 estimates indices based on each node's own observations and only the estimated indices are fused. Thus, as discussed in Table \ref{table_comparison}, DC-OMP 1 incurs   smaller communication overhead compared to DC-OMP 2. As $M$ increases, $P_d$  converges to 1 in all algorithms. This is because when  the number of compressive measurements at each node increases,  OMP (with or without collaboration) works better and recovers the sparsity pattern almost exactly with even a single measurement vector.  In resource constrained distributed networks, especially in sensor networks, it is desirable to perform the desired task by employing  less measurement data  (i.e. with small $M$) at each node distributing the computational complexity among nodes to save the overall node power, and the proposed algorithms are promising in this case compared to performing OMP at each node independently (based on Algorithm \ref{algo0})  followed by  fusion.

\begin{figure}
\centerline{\epsfig{figure=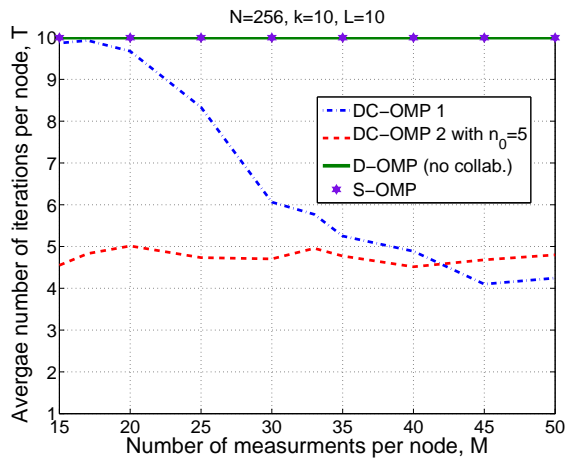,width=8.50cm}}
\caption{Number of iterations required for algorithm to be terminated at each node for Algorithms \ref{algo_DC-OMP1} and \ref{algo1}; $N=256$, $k=10$, $L=10$, $\bar\gamma  = \frac{1}{L}\sum_{l=0}^{L-1} \frac{||\mathbf s_l||^2}{N \sigma_v^2} = 28 dB$ }
  \label{fig_3}
\end{figure}

To further illustrate the efficiency of the proposed algorithms, in Fig. \ref{fig_3}, we plot the average number of iterations required by  DC-OMP 1 and DC-OMP 2  to terminate at  each node. It is observed from Fig. \ref{fig_3} that, DC-OMP 2 conducts a smaller  number of iterations per node compared to sparsity index  ($\approx k/2$) before algorithm terminates for all values of $M$. As $M$ increases, the number of iterations per node required for Algorithm  DC-OMP 1 to terminate  also reduces  considerably  compared to the sparsity index.  Although, as $M$ increases, D-OMP also correctly recovers the sparsity pattern, DC-OMP 1 does it with very few iterations per node.   D-OMP with no collaboration requires $k$ iterations at each node irrespective of the value of $M$.

Next, we demonstrate the performance of the proposed algorithms as the number of nodes in the network, $L$, increases while keeping the number of measurements per node, $M$, at a  fixed value. In Figures \ref{fig_01}-\ref{fig_02}, we plot, the probability of correctly recovering the common sparsity pattern, and the fraction  of the support correctly recovered, respectively, vs $L$. For the local communication phase of  DC-OMP 2 we assume that the neighborhood size for each node as $L/2$.  We set $M=30$, $N=256$, $k=10$ and $\bar\gamma  = \frac{1}{L}\sum_{l=0}^{L-1} \frac{||\mathbf s_l||^2}{N \sigma_v^2} = 28 dB$.

\begin{figure}[htb]
\centerline{\epsfig{figure=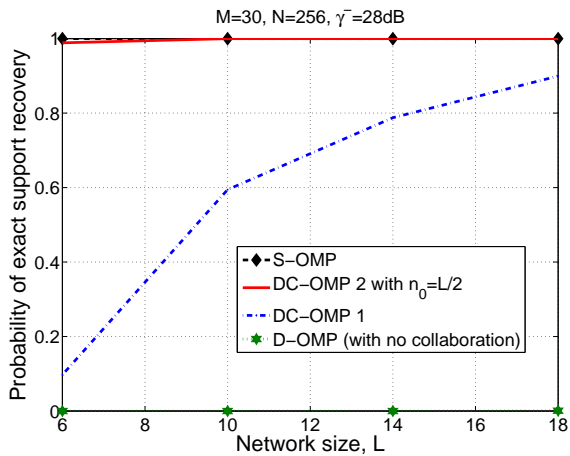,width=8.50cm}}
\caption{Probability of exact  sparsity pattern recovery vs $L$; $N=256$, $k=10$, $M=30$, $\bar\gamma  = \frac{1}{L}\sum_{l=0}^{L-1} \frac{||\mathbf s_l||^2}{N \sigma_v^2} = 28 dB$}
\label{fig_01}
\end{figure}

\begin{figure}[htb]
\centerline{\epsfig{figure=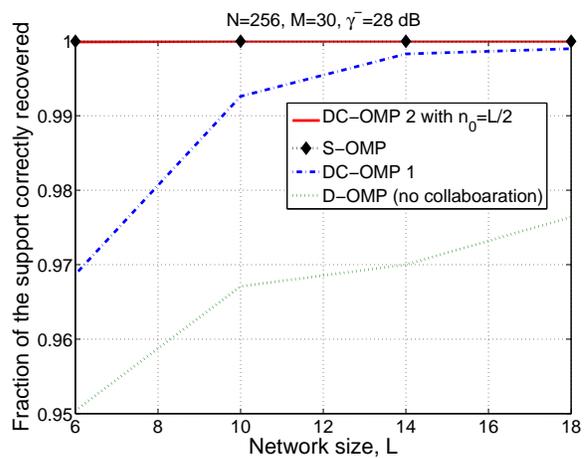,width=8.50cm}}
\caption{Fraction  of the support correctly recovered vs $L$; $N=256$, $k=10$, $M=30$, $\bar\gamma  = \frac{1}{L}\sum_{l=0}^{L-1} \frac{||\mathbf s_l||^2}{N \sigma_v^2} = 28 dB$}
\label{fig_02}
\end{figure}


From Fig \ref{fig_01} and \ref{fig_02}, it can be seen that the performance of both proposed algorithms DC-OMP 1 and DC-OMP 2 improves as $L$ increases. In particular, DC-OMP 2 correctly recovers the sparsity pattern and  yields  similar performance  as  S-OMP with relatively very  small $L$ for a given $M$.  The performance of DC-OMP 1, which performs only index fusion, is improved significantly as $L$ increases for a given small value of $M$ (slightly greater than $k$). The performance of D-OMP with no collaboration does not improve considerably as $L$ increases.  It is noted that in D-OMP (with no collaboration), individual nodes estimate the support by running standard OMP at each node independently, and then the estimated support sets are fused to get a global estimate. Since the fusion is performed after estimating the support sets, the performance of D-OMP is ultimately limited by the number of compressive measurements per node, $M$. However, in DC-OMP 1, estimated indices are fused at each iteration and when the number of nodes increases, more accurate indices for the true support can be estimated by step 4 in DC-OMP 1 in Algorithm \ref{algo1} at each iteration.

\begin{figure}
\centerline{\epsfig{figure=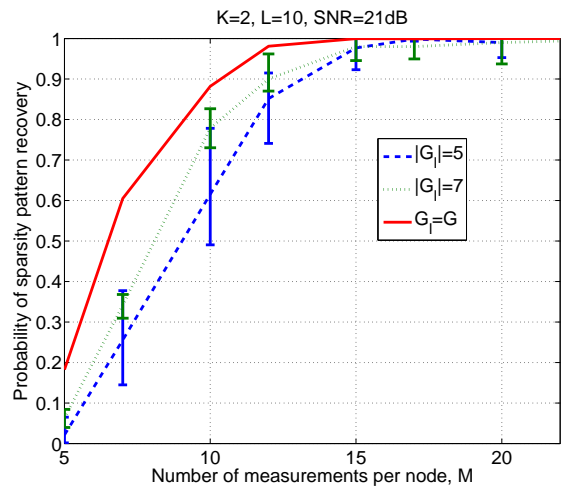,width=8.50cm}}
\caption{Performance of DC-OMP 1 as the size of one hop neighborhood varies; $N=256$, $k=2$, $L=10$, $\bar\gamma  = \frac{1}{L}\sum_{l=0}^{L-1} \frac{||\mathbf s_l||^2}{N \sigma_v^2} = 21 dB$ }
  \label{fig_dcomp2}
\end{figure}


In Figures \ref{fig_1}-\ref{fig_02},  we assumed that $|\mathcal G_l\cup l|=L$ for DC-OMP 1   so that when the algorithm  terminates each node has the same estimated support set. In Fig. \ref{fig_dcomp2}, we compare the performance of DC-OMP 1 when $|\mathcal G_l\cup l|< L$ and the one hop neighborhood size varies. In Fig. \ref{fig_dcomp2}, we plot the average  probability of sparsity pattern recovery as the neighborhood size varies where the average is taken over the nodes. In Fig. \ref{fig_dcomp2}, we let $k=2$, $L=10$ and average SNR $\bar{\gamma}=21 dB$.   We assume that each node has the same number of one-hop neighbors. We also, plot the errorbars which represent the minimum and the maximum deviations of the probability of sparsity pattern recovery over individual nodes. It can be observed that, as the number of one hop neighbors increases, the length of the error bars decreases, thus, each node has almost the same performance in sparsity pattern recovery. Further, it is observed that Algorithm \ref{algo1} when $|\mathcal G_l| \geq k$  converges  with almost the same number of measurements per node as required when $\mathcal G_l=\mathcal G$.

\section{Conclusion} \label{conclusions}
 In this paper, we considered the problem of recovering a common  sparsity pattern  in a  distributed network under communication constraints in  a centralized as well as decentralized manner.  In a centralized setting, the communication constraint is taken into account by employing a shared multiple access channel to forward  observations  at each node to a central fusion center. Then, we showed that under certain conditions, the sparsity pattern can be reliably  recovered  based on the MAC output with  performance that is comparable  to the case when the fusion center has  access to all the observations separately.

We further proposed two decentralized OMP based algorithms for joint  sparsity pattern recovery without depending on a central fusion center. These algorithms  use  minimal communication and computational burden at each node. More specifically, in the proposed algorithms, collaboration and fusion are exploited at different stages of the standard OMP algorithm to have a more accurate estimate for the common sparsity pattern with small number of compressive measurements per node as well as  less communication burden.

 In future work, we will develop decentralized algorithms for simultaneous  approximation of sparse signals with structured and more complex joint sparsity models. We will further investigate the impact of  channel effects on the joint sparse recovery in distributed networks.


\bibliographystyle{IEEEtran}
\bibliography{IEEEabrv,bib1}

\begin{thebibliography}{10}
\providecommand{\url}[1]{#1}
\csname url@samestyle\endcsname
\providecommand{\newblock}{\relax}
\providecommand{\bibinfo}[2]{#2}
\providecommand{\BIBentrySTDinterwordspacing}{\spaceskip=0pt\relax}
\providecommand{\BIBentryALTinterwordstretchfactor}{4}
\providecommand{\BIBentryALTinterwordspacing}{\spaceskip=\fontdimen2\font plus
\BIBentryALTinterwordstretchfactor\fontdimen3\font minus
  \fontdimen4\font\relax}
\providecommand{\BIBforeignlanguage}[2]{{%
\expandafter\ifx\csname l@#1\endcsname\relax
\typeout{** WARNING: IEEEtran.bst: No hyphenation pattern has been}%
\typeout{** loaded for the language `#1'. Using the pattern for}%
\typeout{** the default language instead.}%
\else
\language=\csname l@#1\endcsname
\fi
#2}}
\providecommand{\BIBdecl}{\relax}
\BIBdecl

\bibitem{Mishali1}
M.~Mishali and Y.~Eldar, ``Blind multi-band signal reconstruction: Compressed
  sensing for analog signals,'' \emph{IEEE Trans. Signal Processing}, vol.~57,
  no.~3, pp. 993--1009, Mar. 2009.

\bibitem{Obozinski1}
G.~Obozinski, B.~Taskar, , and M.~Jordan, ``Joint covariate selection and joint
  subspace selection for multiple classification problems,'' \emph{Stat.
  Comput.}, vol.~20, no.~2, pp. 231--252, 2010.

\bibitem{Baron1}
D.~Baron, M.~Duarte, S.~Sarvotham, M.~B. Wakin, and R.~G. Baraniuk,
  ``Distributed compressed sensing,'' \emph{Rice Univ. Dept. Elect. Comput.
  Eng. Houston, TX, Tech. Rep. TREE–0612}, Nov 2006.

\bibitem{ling1}
Q.~Ling and Z.~Tian, ``Decentralized support detection of multiple measurement
  vectors with joint sparsity,'' in \emph{Proc. Acoust., Speech, Signal
  Processing (ICASSP)}, 2011, pp. 2996--2999.

\bibitem{Zeng1}
F.~Zeng, C.~Li, and Z.~Tian, ``Distributed compressive spectrum sensing in
  cooperative multihop cognitive networks,'' \emph{IEEE Journal of Selected
  Topics in Signal Processing}, vol.~5, no.~1, pp. 37--48, Feb. 2011.

\bibitem{Bazerque1}
J.~A. Bazerque and G.~B. Giannakis, ``Distributed spectrum sensing for
  cognitive radio networks by exploiting sparsity,'' \emph{IEEE Trans. Signal
  Processing}, vol.~58, no.~3, pp. 1847--1862, Mar. 2010.

\bibitem{Gorodnitsky2}
I.~F. Gorodnitsky, J.~S. George, and B.~D. Rao, ``Neuromagnetic source imaging
  with focuss: A recursive weighted minimum norm algorithm,'' \emph{J.
  Electroencephalog. Clinical Neurophysiol.}, vol.~95, no.~4, pp. 231--251,
  1995.

\bibitem{Cotter1}
S.~F. Cotter, B.~D. Rao, K.~Engan, and K.~Kreutz-Delgado, ``Sparse solutions to
  linear inverse problems with multiple measurement vectors,'' \emph{IEEE
  Trans. Signal Processing}, vol.~53, no.~7, pp. 2477--2488, July 2005.

\bibitem{Lee2}
O.~Lee, J.~Kim, Y.~Bresler, and J.~Ye, ``Compressive diffuse optical
  tomography: Noniterative exact reconstruction using joint sparsity,''
  \emph{IEEE Trans. Med. Imag.}, vol.~30, no.~5, pp. 1129--1142, May 2011.

\bibitem{Tropp4}
J.~Tropp, A.~Gilbert, and M.~Strauss, ``Algorithms for simultaneous sparse
  approximation. part i: Greedy pursuit,'' \emph{Signal Processing, special
  issue on Sparse approximations in signal and image processing}, vol.~86,
  no.~4, pp. 572--588, 2006.

\bibitem{Tropp3}
------, ``Algorithms for simultaneous sparse approximation. part ii: Convex
  relaxation,'' \emph{Signal Processing, special issue on Sparse approximations
  in signal and image processing}, vol.~86, no.~4, pp. 589--602, 2006.

\bibitem{Chen2}
J.~Chen and X.~Huo, ``Theoretical results on sparse representations of
  multiple-measurement vectors,'' \emph{IEEE Trans. Signal Processing},
  vol.~54, no.~12, pp. 4634--4643, Dec. 2006.

\bibitem{Obozinski2}
G.~Obozinski, M.Wainwright, and M.~Jordan, ``Support union recovery in
  high-dimensional multivariate regression,'' \emph{Ann. Stat.}, vol.~39,
  no.~1, pp. 1--47, 2011.

\bibitem{Wipf2}
D.~Wipf and B.~Rao, ``An empirical bayesian strategy for solving the
  simultaneous sparse approximation problem,'' \emph{IEEE Trans. Signal
  Processing}, vol.~55, no.~7, pp. 3704--3716, July 2007.

\bibitem{Eldar4}
Y.~C. Eldar and H.~Rauhut, ``Average case analysis of multichannel sparse
  recovery using convex relaxation,'' \emph{IEEE Trans. Inform. Theory},
  vol.~56, no.~1, pp. 505--519, Jan. 2010.

\bibitem{Eldar1}
Y.~C. Eldar and M.~Mishali, ``Robust recovery of signals from a structured
  union of subspaces,'' \emph{IEEE Trans. Information Theory}, vol.~55, no.~11,
  pp. 5302--5316, Nov. 2009.

\bibitem{mergen1}
G.~Mergen and L.~Tong, ``Type based estimation over multiaccess channels,''
  \emph{IEEE Trans. Signal Processing}, vol.~54, no.~2, pp. 613--626, Feb.
  2006.

\bibitem{li1}
W.~Li and H.~Dai, ``Distributed detection in wireless sensor networks using a
  multiple access channel,'' \emph{IEEE Trans. Signal Processing.}, vol.~55,
  no.~3, pp. 822--833, March 2007.

\bibitem{tropp1}
J.~Tropp and A.~Gilbert, ``Signal recovery from random measurements via
  orthogonal matching pursuit,'' \emph{IEEE Trans. Inform. Theory}, vol.~53,
  no.~12, pp. 4655--4666, Dec. 2007.

\bibitem{Tropp2}
J.~A. Tropp, A.~C. Gilbert, and M.~J. Strauss, ``Simultaneous sparse
  approximation via greedy pursuit,'' in \emph{Proc. Acoust., Speech, Signal
  Processing (ICASSP)}, 2005, pp. V--721--V--724.

\bibitem{thakshila_icassp1}
T.~Wimalajeewa and P.~K. Varshney, ``Cooperative sparsity pattern recovery in
  distributed networks via distributed-\textsc{OMP},'' in \emph{Proc. Acoust.,
  Speech, Signal Processing (ICASSP)}, Vancouver, Canada, May 2013.

\bibitem{Kim1}
J.~M. Kim, O.~K. Lee, and J.~C. Ye, ``Compressive \textsc{MUSIC}: Revisiting
  the link between compressive sensing and array signal processing,''
  \emph{IEEE Trans. Inform. Theory}, vol.~58, no.~1, pp. 278--301, Jan. 2012.

\bibitem{Ling2}
Q.~Ling and Z.~Tian, ``Decentralized sparse signal recovery for compressive
  sleeping wireless sensor networks,'' \emph{IEEE Trans. Signal Processing},
  vol.~58, no.~7, pp. 3816--3827, July 2010.

\bibitem{Rabbat1}
M.~Rabbat, J.~D. Haupt, A.~Singh, and R.~D. Nowak, ``Decentralized compression
  and predistribution via randomized gossiping,'' in \emph{Int. Workshop on
  Info. Proc. in Sensor Networks (IPSN)}, Nashville, TN, Apr. 2006.

\bibitem{Patterson1}
S.~Patterson, Y.~C. Eldar, and I.~Keidar, ``Distributed sparse signal recovery
  for sensor networks,'' \emph{[Online] Available:
  http://arxiv.org/pdf/1212.6009v2.pdf, preprint}, 2013.

\bibitem{Sundman_icassp1}
D.~Sundman, S.~Chatterjee, and M.~Skoglund, ``A greedy pursuit algorithm for
  distributed compressed sensing,'' in \emph{Proc. Acoust., Speech, Signal
  Processing (ICASSP)}, Kyoto, Japan, March 2012, p. 2729–2732.

\bibitem{Sundman_arxiv1}
------, ``Distributed greedy pursuit algorithms,'' \emph{[Online] Available:
  http://arxiv.org/pdf/1306.6815v2.pdf, preprint}, 2013.

\bibitem{Huang_arxiv1}
H.~Huang, S.~Misra, W.~Tang, H.~Barani, and H.~Al-Azzawi1, ``Applications of
  compressed sensing in communications networks,'' \emph{[Online] Available:
  http://arxiv.org/pdf/1305.3002v3.pdf, preprint}, 2014.

\bibitem{Fletcher2}
A.~Fletcher and S.~Rangan, ``Orthogonal matching pursuit from noisy random
  measurements: A new analysis,'' in \emph{NIPS'09}, 2009.

\bibitem{Kutyniok_arxiv1}
G.~Kutyniok, ``Theory and applications of compressed sensing,'' \emph{arXiv
  preprint arXiv:1203.3815}, 2012.

\bibitem{Eldar3}
Y.~C. Eldar, P.~Kuppinger, and H.~Bolcskei, ``Block-sparse signals: Uncertainty
  relations and efficient recovery,'' \emph{IEEE Trans. Signal Processing},
  vol.~58, no.~6, pp. 3042--3054, June 2010.

\bibitem{cover1}
T.~M. Cover and J.~A. Thomas, \emph{Elements of Information Theory}.\hskip 1em
  plus 0.5em minus 0.4em\relax John Wiley and Sons, Inc., NY, 2006.

\bibitem{wain2}
M.~J. Wainwright, ``Information-theoretic limits on sparsity recovery in the
  high-dimensional and noisy setting,'' \emph{IEEE Trans. Inform. Theory},
  vol.~55, no.~12, pp. 5728--5741, Dec. 2009.

\end{thebibliography}

\newpage

\end{document}